\documentclass[12pt]{article}

\usepackage{hyperref}

\usepackage{amsfonts, amsmath, amssymb, amsthm}

\usepackage{mathtools}

\usepackage{tikz}
\usepackage{a4}
\usepackage{graphicx}

\usepackage{bbold}

 \usepackage{xcolor}

 \definecolor{Refkey}{RGB}{255,127,0}
 \definecolor{Labelkey}{RGB}{127,0,255}
 \makeatletter 
  \def\SK@refcolor{\color{Refkey}}
  \def\SK@labelcolor{\color{Labelkey}}
 \makeatother

  \definecolor{mdg}{RGB}{0,177,0} 
  \definecolor{mdb}{RGB}{0,0,191}
  \definecolor{mddb}{RGB}{0,0,91}
  \definecolor{mdy}{RGB}{255,69,0} 
  \definecolor{gray}{RGB}{99,99,99} 
  \definecolor{darkgreen}{RGB}{0,128,0} 
  \definecolor{darkblue}{RGB}{0,0,128}

\DeclareMathOperator{\Gr}{Gr}
\DeclareMathOperator{\oplrcorner}{\lrcorner}

\newtheorem{theorem}{Theorem}
\newtheorem{proposition}{Proposition}

\theoremstyle{definition}
\newtheorem{assumption}{Assumption}
\newtheorem{convention}{Convention}
\newtheorem{example}{Example}

\theoremstyle{remark}

\title{Grassmannian-parameterized solutions to direct-sum polygon and simplex equations}
\author{Aristophanes Dimakis \and Igor Korepanov}
\date{October 2020}

\begin{document}

\sloppy

\maketitle

\begin{abstract}
We consider polygon and simplex equations, of which the simplest nontrivial examples are pentagon (5-gon) and Yang--Baxter (2-simplex), respectively. We examine the general structure of $(2n+1)$-gon and $2n$-simplex equations in direct sums of vector spaces. Then we provide a construction for their solutions, parameterized by elements of the Grassmannian~$\Gr(n+1,2n+1)$.
\end{abstract}

\section{Introduction}

\subsection{Generalities about set-theoretic polygon and simplex equations}

Polygon and simplex equations are two families of algebraic equations that naturally arise, and find applications, in mathematical physics, topology, representation theory; there is also interesting combinatorics related to these equations. A short but comprehensive enough historical review, with many references, can be found in the Introduction to paper~\cite{DM-H}; here we just mention some aspects of our interest in the present paper.

The first nontrivial example of polygon equations is \emph{pentagon} equation
\begin{equation}\label{5g}
A^{(1)}_{12} A^{(3)}_{13} A^{(5)}_{23} =  A^{(4)}_{23} A^{(2)}_{12},
\end{equation}
below we also represent it graphically in Figure~\ref{fig:pentagon}. Concerning the \emph{superscripts} in~\eqref{5g}, we hope to explain below why we choose them like that; right now it is enough to say that they symbolize that $A^{(q)}$ with different~$q$'s may be different objects.

In the ``set-theoretic'' interpretation of~\eqref{5g}, the \emph{subscripts} represent the \emph{sets} where \emph{mappings}~$A^{(q)}$ act. Usually, these sets carry some algebraic structure, for instance, of a group or a quandle (the latter also called ``distributive groupoid''~\cite{Matveev,Kamada}). In~\eqref{5g}, three sets $X_1$, $X_2$ and~$X_3$ are involved; often---but not necessarily---they are copies of one fixed set~$X$. Each mapping $A^{(q)}_{ij}$ acts in the corresponding direct product:
\[
A^{(q)}_{ij}\colon\quad X_i\times X_j \to X_i\times X_j.
\]
Also, $A^{(q)}_{ij}$ is identified with the direct product of itself and the \emph{identity} mapping~$\mathbb 1_k$ in the lacking set~$X_k$, so the whole l.h.s.\ and r.h.s.\ of~\eqref{5g} act in $X_1 \times X_2 \times X_3$.

Similarly, the first nontrivial example of \emph{simplex} equations is the \emph{Yang--Baxter} equation
\begin{equation}\label{YB}
R^{(1)}_{12} R^{(2)}_{13} R^{(3)}_{23} =  R^{(3)}_{23} R^{(2)}_{13} R^{(1)}_{12},
\end{equation}
then goes \emph{Zamolodchikov tetrahedron} equation,
\begin{equation}\label{T}
R^{(1)}_{123} R^{(2)}_{145} R^{(3)}_{246} R^{(4)}_{356} =  R^{(4)}_{356} R^{(3)}_{246} R^{(2)}_{145} R^{(1)}_{123},
\end{equation}
then ``4-simplex equation'', and so on. In the set-theoretic setting, they are interpreted in the same way as described above. For instance, in~\eqref{T}, each~$R$ maps the direct product of three sets into itself:
\[
R^{(q)}_{ijk}\colon\quad X_i\times X_j\times X_k \to X_i\times X_j\times X_k,
\]
and is also identified when necessary with its direct product with the identity mappings in the lacking sets.

There are many interrelations between polygon and simplex equations. Of these, we mention here, and will exploit in this paper, the ``three-color decomposition'' of simplex equations. This means that, assuming a special form (see~\eqref{RABP} below) of solutions, the simplex equation breaks into three independent parts, two of which are nothing but polygon equations, while the third can be treated as a ``consistency condition'' between these two. Such decomposition first appeared, for pentagon and 4-simplex, in~\cite{KashS}, and was generalized to other equations and analyzed combinatorially in detail in~\cite{DM-H}. In some cases, it allows to construct simplex solutions \emph{from} polygon solutions.

\subsection{Direct-sum equations and motivation for their study}

Our ``direct-sum'' equations are a particular case of set-theoretic equations. Namely, each $X_i=V_i$ is a \emph{vector space} over a field~$F$; accordingly $A^{(q)}_{ij}\colon\; V_i\oplus V_j \to V_i\oplus V_j$ are \emph{linear operators}, and everything happens in the \emph{direct sum} of spaces~$V_i$. Field~$F$ is allowed to be \emph{finite}, as well as all~$V_i$.

Direct-sum \emph{Yang--Baxter} equations have been used in \textbf{knot theory} for quite a long time, although not necessarily under this name. Namely, the famous \emph{tricolorability}~\cite[p.~8]{GP} property (not to be confused with the three-color decomposition dealt with below in this paper!) that allows to distinguish the trefoil knot from the unknot uses the rule of coloring knot fragments that can be described in terms of the direct-sum Yang--Baxter, with the field $F=\mathbb F_3$ being the Galois field of three elements. More generally, the \emph{Alexander quandle}~\cite{Kamada}, considered over field~$F$, gives a coloring rule of a direct-sum Yang--Baxter kind, which can be written as
\[
\begin{pmatrix} x & y \end{pmatrix} \mapsto \begin{pmatrix} x & y \end{pmatrix} \begin{pmatrix} 1 & 1-t \\ 0 & t \end{pmatrix},
\]
where $t\in F$\/ is a parameter (note that we thus prefer in this paper to regard matrices as acting on \emph{rows} rather than columns).

Similarly to how Yang--Baxter is applied to knot theory, set-theoretic---and in particular direct-sum---solutions to \emph{polygon} equations can be expected to find their applications to piecewise linear manifold invariants. That is because they correspond naturally to \emph{Pachner moves}---elementary re-buildings of a triangulation of a given manifold. For instance, pentagon corresponds to Pachner move 2--3 in three dimensions.

Returning to Yang--Baxter, one more interesting area where its direct-sum solutions have appeared is the tropical limit of Kadomtsev--Petviashvili and Korteweg--deVries solitonic equations, see~\cite[formulas (5.6), (5.7)]{DM-H-tropical}. 

More general direct-sum simplex equations---generalizations of Yang--Baxter---appeared in Hietarinta's paper~\cite{Hietarinta}, where they were used for constructing ``permutation-type'' solutions to \emph{quantum} equations. Permutation-type solutions can be expected to be the first step in constructing more sophisticated solutions to higher simplex equations.

\subsection{Constant and nonconstant solutions}

Sets~$X_i$ introduced above are often \emph{naturally isomorphic} to one another, that is, $X_1,X_2,\ldots$ are just copies of one set~$X$; the same applies of course to vector spaces~$V_1,V_2,\ldots$\,. This allows us to speak of the case where our mappings, or linear operators, such as $A^{(q)}$ or~$R^{(q)}$, are also copies of one and the same mapping/operator for all~$q$'s. In this case we call the considered equation \emph{constant}. Constant equations are certainly of interest: enough to say that it is exactly the constant Yang--Baxter that finds its applications in \emph{knot theory} and its generalizations~\cite{CS} (the same applies to the quantum case~\cite{ADW}). Also, Hietarinta in his paper~\cite{Hietarinta} classified solutions to constant Yang--Baxter and higher simplex equations.

If mappings/operators are allowed to be different, the equation is called \emph{nonconstant}. We said already that this is exactly the case in our equations such as \eqref{5g} or~\eqref{YB}. So, our paper is about the \emph{non-constant} case.

\subsection{What is done in the main part of the paper}

Speaking more specifically about our linear operators, they will depend on parameters that can be thought of together as an element of the \emph{Grassmannian}~$\Gr(n+1,2n+1)$---an $(n+1)$-dimen\-sional plane in a $(2n+1)$-dimen\-sional linear space, $n=1,2,\ldots$\,. This construction provides solutions to $(2n+1)$-gon and $2n$-simplex equations. In the trivial but still instructive case $n=1$, these are \emph{trigon} and Yang--Baxter, see Example~\ref{x:n=1} below; for $n=2$, these are \emph{pentagon} and 4-simplex, and so on. All our solutions depend only on \emph{ratios of Pl\"ucker coordinates} of the Grassmannian, see \eqref{A-def} and~\eqref{B-def} below. Our ``polygons'' here have thus always an odd number~$(2n+1)$ of ``vertices'', so we leave for future work half of the possible equations; no doubt that they are interesting and deserve a separate research.

Below, in Section~\ref{s:structure}, we investigate the general structure of $(2n+\nobreak 1)$-gon and $2n$-simplex equations in direct sums of vector spaces. In particular, we describe the \emph{three-color decomposition}~\cite{DM-H} of the $2n$-simplex equations that permits in some cases to obtain their solutions from solutions of $(2n+1)$-gon equations. Although we do it for our direct-sum case, the reader can notice that the combinatorial structure will be the same for either general set-theoretical or quantum case. We show also how to obtain solutions to $(2n-1)$-simplex equations using a reduction procedure.

Then, in Section~\ref{s:sols}, we provide our Grassmannian-based solutions. The proof that they are indeed solutions for any~$n$ combines Grassmann algebra and Pl\"ucker bilinear relations.

Finally, in Section~\ref{s:d}, we discuss our results and possible directions for further research.

\section[$(2n+1)$-gon and $2n$-simplex equations in direct sums]{$\boldsymbol{(2n+1)}$-gon and $\boldsymbol{2n}$-simplex equations in direct sums}\label{s:structure}

\subsection{General notions}\label{ss:gn}

Polygon and simplex equations in direct sums of vector spaces are equations on linear operators. Denote these spaces, for a chosen equation,~$V_i$, where index~$i$ runs over some (finite) set~$\mathcal C \ni i$. The equation lives hence in
\begin{equation}\label{VC}
V_{\mathcal C} \, \stackrel{\mathrm{def}}{=} \, \bigoplus_{i\in \mathcal C} V_i.
\end{equation}

Each individual operator, however, acts nontrivially only in \emph{some} spaces. That is, for such operator~$A$ there is a subset~$\mathcal B \subset \mathcal C$ such that both spaces $V_{\mathcal B}$ and $V_{\mathcal C \setminus \mathcal B}$ (defined in the same way as in~\eqref{VC}) are invariant for~$A$ and, moreover, $A$ acts as an identity operator in~$V_{\mathcal C \setminus \mathcal B}$:
\[
A = A|_{\mathcal B} \oplus \mathbb 1_{\mathcal C \setminus \mathcal B}.
\]
In this situation, we identify $A$ with its restriction~$A|_{\mathcal B}$.

In this paper, we consider the simplest case where each~$V_i$ is a \emph{one-dimen\-sional} space over a field~$F$, and has a \emph{fixed basis}. This means of course that we have identified each~$V_i$ with~$F$, and also our operators are naturally identified with matrices.

\subsection[$(2n+1)$-gon equation]{$\boldsymbol{(2n+1)}$-gon equation}\label{ss:gon}

The direct sum $(2n+1)$-gon equation considered here is defined on a vector space~$V_{(2n+1)\textrm{-gon}}$ of dimension $\binom{n}{2} = \frac{n(n+1)}{2}$ over a field~$F$. We think of it, in accordance with Subsection~\ref{ss:gn}, as a direct sum of one-dimen\-sional spaces~$V_i$, where $i$ takes values in the set
\[
\mathcal C = [n(n+1) / 2] \stackrel{\mathrm{def}}{=} \{ 1,2,\dots, n(n+1) / 2 \}.
\]
Explicitly, we write elements of~$V_{(2n+1)\textrm{-gon}}$ as $\frac{n(n+1)}{2}$-\emph{row} vectors, and the $(2n+1)$-gon equation is
\begin{equation}\label{(2n+1)-gon}
    A_{\mathcal{B}_1}^{(1)}A_{\mathcal{B}_3}^{(3)}\cdots  A_{\mathcal{B}_{2n+1}}^{(2n+1)}=A_{\mathcal{B}_{2n}}^{(2n)}A_{\mathcal{B}_{2n-2}}^{(2n-2)}\cdots A_{\mathcal{B}_2}^{(2)} .
\end{equation}
Here each $A^{(q)}$, \ $q=1,2,\ldots  ,2n+1$, is a linear operator acting nontrivially in the direct sum of~$n$ one-dimen\-sional spaces and thus identified, according to Subsection~\ref{ss:gn}, with an $n\times n$-matrix; as our space~$V_{(2n+1)\textrm{-gon}}$ consists of row vectors, the matrices act on the \emph{right}. Each $\mathcal{B}_q$ in~\eqref{(2n+1)-gon} is the set of numbers of spaces, or simply \emph{positions} in the row where $A^{(q)}$ acts nontrivially:
\begin{equation}\label{Bq}
\mathcal{B}_q=\{ b_{q,1},\ldots ,b_{q,n} \}.
\end{equation}
Below, we are going to specify these~$\mathcal{B}_q$ in a way that may seem somewhat roundabout, but is actually very convenient for our aims here.

\begin{convention}\label{c:o}
In this paper, we list the elements of all sets always in the \emph{increasing} order, if not explicitly stated otherwise. For instance, $b_{q,1} < \dots < b_{q,n}$ in~\eqref{Bq}.
\end{convention}

We introduce two sequences of pairs of numbers from $\{1,2,\dots,2n+1\}$ called \emph{initial} and \emph{final} sequence. Initial sequence consists of pairs (odd number, even number), taken in the lexicographic order, that is,
\begin{equation}\label{gon}
 \begin{aligned}
    (1,2), & (1,4),\ldots,(1,2n),\, \ldots,\, \\
   & \boldsymbol{(2k-1,2k),\ldots,(2k-1,2n)},\, 
   \ldots ,\, (2n-1,2n) .
 \end{aligned}
\end{equation}
Here we highlighted in bold a ``typical'' subsequence consisting of pairs containing $2k-1$. Similarly, final sequence consists of pairs (even number, odd number), taken also in the lexicographic order, that is,
\begin{equation}\label{rgon}
 \begin{aligned}
   (2,3), & (2,5),\ldots  (2,2n+1),\, \ldots ,\, \\
    & \boldsymbol{(2k,2k+1),\ldots,(2k,2n+1)},\, 
    \ldots ,\, (2n,2n+1) ,
 \end{aligned}
\end{equation}
where we highlighted in bold a ``typical'' subsequence consisting of pairs containing $2k$. Both sequences clearly have length~$\frac{n(n+1)}{2}$.

\textbf{By definition, }$\mathcal B_q$ consists of such positions~$b=b_{q,i}$ for which $q$ is present in the $b$'s pair at least in one of the sequences \eqref{gon} and~\eqref{rgon}. Clearly, there are exactly~$n$ such positions for any~$q=1,2,\dots,2n+1$. Explicit formulas for~$\mathcal B_q$ are given in Subsection~\ref{ss:Bcal}.

We introduce the following \emph{indexing rule} for input and output $n$-rows for matrices~$A^{(q)}$. For a given~$q$, define the set
\begin{equation}\label{Lq}
    L_q=\{1,2,\ldots,2n+1\}\setminus\{q\} \stackrel{\mathrm{def}}{=} \{p_1,p_2,\ldots,p_{2n}\},
\end{equation}
then the entries of the mentioned rows are indexed as follows:
\begin{equation}\label{gon-inds}
   \begin{pmatrix} u_{p_1,q} & u_{p_3,q} & \dots & u_{p_{2n-1},q} \end{pmatrix} A^{(q)} = \begin{pmatrix} u_{p_2,q} & u_{p_4,q} & \dots & u_{p_{2n},q} \end{pmatrix}.
\end{equation}
Here we identify if needed $(p_i,q)$ with~$(q,p_i)$, this cannot bring confusion.

\begin{proposition}\label{p:ind-rule}
The indexing rule~\eqref{gon-inds} is consistent with both sides of the $(2n+1)$-gon equation~\eqref{(2n+1)-gon}. Namely, take an input row vector for both sides of~\eqref{(2n+1)-gon} whose entries are indexed according to the sequence~\eqref{gon}:
\begin{equation}\label{ini-row}
\begin{pmatrix} u_{1,2} & u_{1,4} & \dots & u_{2n-1, 2n} \end{pmatrix}.
\end{equation}
Then, as we apply in turn the matrices~$A^{(q)}$ in either l.h.s.\ or r.h.s.\ of~\eqref{(2n+1)-gon} and change at each step the indexing of the corresponding $n$ entries according to~\eqref{gon-inds},
 \begin{itemize}
  \item[$\mathrm{(i)}$] at each step, there are exactly~$n$ index pairs in the input $\frac{n(n+1)}{2}$-row fitting~\eqref{gon-inds},
  \item[$\mathrm{(ii)}$] the final obtained sequence coincides with~\eqref{rgon}.
 \end{itemize}
\end{proposition}

\begin{proof}
 \begin{itemize}
  \item[$\mathrm{(i)}$] Indeed, consider the l.h.s. for concreteness. There are of course exactly~$n$ index pairs in~\eqref{ini-row} fitting~$A^{(1)}$. After applying~$A^{(1)}$, there appear $u_{1,3}, u_{1,5}, \ldots $, so, in particular, there are now $n$ index pairs fitting~$A^{(3)}$, and so on.
  \item[$\mathrm{(ii)}$] In the l.h.s., an entry~$u_{2k-1,2l}$ is first replaced with~$u_{2k-1,2l+1}$ by~$A^{(2k-1)}$ and then with~$u_{2k,2l+1}$ by~$A^{(2l+1)}$, as required. In the r.h.s., an entry~$u_{2k-1,2l}$ is first replaced with~$u_{2k,2l}$ by~$A^{(2l)}$ and then with~$u_{2k,2l+1}$ by~$A^{(2k)}$, except for the case $k=l$, where only one matrix, $A^{(2k)}$, takes part in the play, and replaces~$u_{2k-1,2k}$ directly with~$u_{2k,2k+1}$.
 \end{itemize}
\end{proof}

\begin{example}
Figure~\ref{fig:pentagon} shows what happens with index pairs for $n=2$, that is, \emph{pentagon} equation.
\end{example}

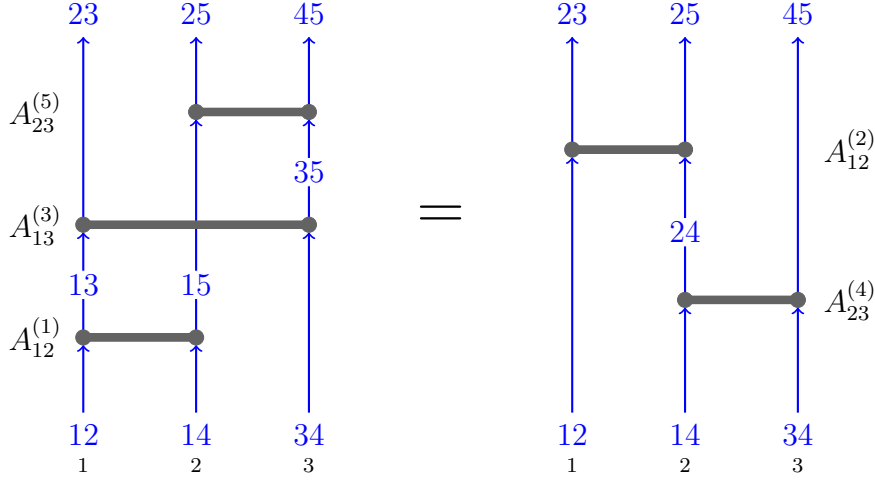
\begin{figure}[t]
 \begin{center}
  \begin{tikzpicture}

\draw[blue, thick, -> ] (0,0) -- (0,0.9) ;
\draw[blue, thick, -> ] (1.5,0) -- (1.5,0.9) ;
\draw[blue, thick, -> ] (3,0) -- (3,2.4) ;

\draw[blue, thick, -> ] (0,1) -- (0,2.4) ;
\draw[blue, thick, -> ] (1.5,1) -- (1.5,3.9) ;
\draw[blue, thick, -> ] (3,2.5) -- (3,3.9) ;

\draw[blue, thick, -> ] (0,2.5) -- (0,5) ;
\draw[blue, thick, -> ] (1.5,4) -- (1.5,5) ;
\draw[blue, thick, -> ] (3,4) -- (3,5) ;

\draw (0,1) node[black, anchor=east] {$A_{12}^{(1)}$ \hspace{ 1em }} ;
\filldraw[gray] (0,0.95) rectangle (1.5,1.05) ;
\filldraw[gray] (0,1) circle (0.1);
\filldraw[gray] (1.5,1) circle (0.1);

\draw (0,2.5) node[black, anchor=east] {$A_{13}^{(3)}$ \hspace{ 1em }} ;
\filldraw[gray] (0,2.45) rectangle (3,2.55) ;
\filldraw[gray] (0,2.5) circle (0.1);
\filldraw[gray] (3,2.5) circle (0.1);

\draw (0,4) node[black, anchor=east] {$A_{23}^{(5)}$ \hspace{ 1em }} ;
\filldraw[gray] (1.5,3.95) rectangle (3,4.05) ;
\filldraw[gray] (1.5,4) circle (0.1);
\filldraw[gray] (3,4) circle (0.1);

\filldraw[white] (0,1.7) circle (1ex);
\draw[blue] (0,1.7) node {$13$};

\filldraw[white] (1.5,1.7) circle (1ex);
\draw[blue] (1.5,1.7) node {$15$};

\filldraw[white] (3,3.2) circle (1ex);
\draw[blue] (3,3.2) node {$35$};

\draw (0,-0.3) node[blue] {$12$} ;
\draw (1.5,-0.3) node[blue] {$14$} ;
\draw (3,-0.3) node[blue] {$34$} ;

\draw (0,5.3) node[blue] {$23$} ;
\draw (1.5,5.3) node[blue] {$25$} ;
\draw (3,5.3) node[blue] {$45$} ;

\draw (0,-0.7) node {$\scriptstyle 1$} ;
\draw (1.5,-0.7) node {$\scriptstyle 2$} ;
\draw (3,-0.7) node {$\scriptstyle 3$} ;

\draw (4.75,2.6) node {\huge =};

\draw[blue, thick, -> ] (6.5,0) -- (6.5,3.4) ;
\draw[blue, thick, -> ] (8,0) -- (8,1.4) ;
\draw[blue, thick, -> ] (9.5,0) -- (9.5,1.4) ;

\draw[blue, thick, -> ] (6.5,3.5) -- (6.5,5) ;
\draw[blue, thick, -> ] (8,1.5) -- (8,3.4) ;
\draw[blue, thick, -> ] (8,3.5) -- (8,5) ;
\draw[blue, thick, -> ] (9.5,1.5) -- (9.5,5) ;

\draw (9.5,1.5) node[black, anchor=west] {\hspace{ .15em } $A_{23}^{(4)}$} ;
\filldraw[gray] (8,1.45) rectangle (9.5,1.55) ;
\filldraw[gray] (8,1.5) circle (0.1);
\filldraw[gray] (9.5,1.5) circle (0.1);

\draw (9.5,3.5) node[black, anchor=west] {\hspace{ .15em } $A_{12}^{(2)}$} ;
\filldraw[gray] (6.5,3.45) rectangle (8,3.55) ;
\filldraw[gray] (6.5,3.5) circle (0.1);
\filldraw[gray] (8,3.5) circle (0.1);

\filldraw[white] (8,2.4) circle (1ex);
\draw[blue] (8,2.4) node {$24$};

\draw (6.5,-0.3) node[blue] {$12$} ;
\draw (8,-0.3) node[blue] {$14$} ;
\draw (9.5,-0.3) node[blue] {$34$} ;

\draw (6.5,5.3) node[blue] {$23$} ;
\draw (8,5.3) node[blue] {$25$} ;
\draw (9.5,5.3) node[blue] {$45$} ;

\draw (6.5,-0.7) node {$\scriptstyle 1$} ;
\draw (8,-0.7) node {$\scriptstyle 2$} ;
\draw (9.5,-0.7) node {$\scriptstyle 3$} ;

\end{tikzpicture}
 \end{center}
 \caption{Pentagon equation. Action of l.h.s.\ or r.h.s.\ on row vectors corresponds to going from the bottom to the top in this figure. The left-hand side is to be compared with the blue part of Figure~\ref{fig:4-simplex}}
 \label{fig:pentagon}
\end{figure}

If we take the inverses of both sides of~\eqref{(2n+1)-gon}, and interchange its l.h.s.\ and r.h.s., we get equation
\begin{equation}\label{gon-inv}
    B_{\mathcal{B}_2}^{(2)} B_{\mathcal{B}_{4}}^{(4)}\cdots B_{\mathcal{B}_{2n}}^{(2n)} = B_{\mathcal{B}_{2n+1}}^{(2n+1)} B_{\mathcal{B}_{2n-1}}^{(2n-1)}\cdots B_{\mathcal{B}_1}^{(1)} ,
\end{equation}
where
\begin{equation}\label{BA}
    B^{(q)} = \left( A^{(q)} \right)^{-1}.
\end{equation}
The natural indexation of row entries for~\eqref{gon-inv} is obtained from~\eqref{gon-inds} by interchanging the rows in its l.h.s.\ and r.h.s.:
\begin{equation}\label{g-i-B}
   \begin{pmatrix} u_{p_2,q} & u_{p_4,q} & \dots & u_{p_{2n},q} \end{pmatrix} B^{(q)} = \begin{pmatrix} u_{p_1,q} & u_{p_3,q} & \dots & u_{p_{2n-1},q} \end{pmatrix} .
\end{equation}

As we will see, having solutions to~\eqref{(2n+1)-gon} \emph{together} with solutions to~\eqref{gon-inv} turns out to be important for constructing solutions to $2n$-simplex equation. We would like to note here that \emph{transposed} matrices $B^{(q)} = ( A^{(q)} )^{\mathrm T}$ also satisfy~\eqref{gon-inv}. In this paper, however, we will be working with the case~\eqref{BA}.

\subsection[$2n$-simplex equation]{$\boldsymbol{2n}$-simplex equation}\label{ss:2n-simplex}

The direct sum $(2n)$-simplex has the form
\begin{equation}\label{2n-simplex}
    R_{\mathcal{A}_1}^{(1)}R_{\mathcal{A}_2}^{(2)}\cdots  R_{\mathcal{A}_{2n+1}}^{(2n+1)}=R_{\mathcal{A}_{2n+1}}^{(2n+1)}\cdots  R_{\mathcal{A}_2}^{(2)}R_{\mathcal{A}_1}^{(1)}
\end{equation}
Here $R^{(q)}$, $q=1,2,\ldots  ,2n+1$ are $(2n)\times (2n)$-matrices. The equation is defined on a vector space of dimension $\binom{2n+1}{2} = n(2n+1)$ consisting of rows which we write as $\begin{pmatrix} u_{1,2} & u_{1,3} & \dots & u_{2n, 2n+1}\end{pmatrix}$. The indices here go in lexicographic order, and in a more detailed form they are
\begin{equation}
 \begin{aligned}
    (1,2), & (1,3),\ldots  ,(1,2n+1),\,\ldots  ,\, \\
    & \boldsymbol{(k,k+1),\ldots,(k,2n+1)},\,\ldots ,\,(2n,2n+1) ,     \label{sim}
 \end{aligned}
\end{equation}
where we again highlighted a ``typical subsequence'' like we did in \eqref{gon} and~\eqref{rgon}. In contrast with the $(2n+1)$-gon equation of Subsection~\ref{ss:gon}, the double indexing~\eqref{sim} is the same for the input and output positions where each matrix~$R^{(q)}$ acts. Namely, these are the positions where $q$ is met in the double index; we can write it in the following way:
\begin{equation}\label{row-Rq}
    \begin{pmatrix} u_{p_1,q} & u_{p_2,q} & \dots & u_{p_{2n},q}\end{pmatrix} R^{(q)} = \begin{pmatrix} v_{p_1,q} & v_{p_2,q} & \dots & v_{p_{2n},q}\end{pmatrix},
\end{equation}
identifying if needed $(p_i,q)$ with~$(q,p_i)$.

Sets $\mathcal{A}_q$ consist, correspondingly, of the positions of double indices in~\eqref{row-Rq} in the sequence~\eqref{sim}. We denote these positions as follows:
\begin{equation}\label{Aq}
\mathcal{A}_q = \{a_{q,1},\ldots  ,a_{q,2n} \},
\end{equation}
similarly to~\eqref{Bq}. Concerning explicit formulas for~$a_{q,i}$, see Proposition~\ref{p:A} below.

\subsection{Positions where matrices act nontrivially}\label{ss:Bcal}

\begin{proposition}\label{p:A}
Positions where matrices act nontrivially in the $2n$-simplex equation are given by
 \begin{equation}\label{A}
    a_{k,j}=\left\{
    \begin{array}{ccc}
        \dfrac{(4n-k)(k-1)}{2}+j & \text{ for} & j\geq k , \\[1.5ex]
        a_{j,k-1} & \text{ for} & j<k .
    \end{array} \right.
 \end{equation}
\end{proposition}

Formula~\eqref{A} first appeared, without proof, in~\cite[p.~16]{DM-H}.

\begin{proof}
 Case $j\geq k$\,. For such~$j$, as one can readily see from~\eqref{sim}, $a_{k,j}=r_k+j$, where
\begin{align*}
 & r_1=0,\qquad r_2=2n-1,\qquad r_3=r_2+(2n-2),\qquad \dots, \\
 & r_k=(2n-1)+(2n-2)+\dots+(2n-k+1) = \dfrac{(4n-k)(k-1)}{2}.
\end{align*}

 Case $j<k$\,. Indeed, in this case pair $j,k$ in~\eqref{sim} corresponds to both $a_{k,j}$ and~$a_{j,k-1}$; see also Example~\ref{x:jk} below.
\end{proof}

\begin{example}\label{x:jk}
To illustrate Proposition~\ref{p:A}, take $n=2$, that is, the 4-simplex equation. Here are the rows consisting of index pairs corresponding to~$\mathcal A_k$; those pairs where $j<k$ are highlighted with color:
\[
\begin{array}{c|cccc}
 & j{=}1 & j{=}2 & j{=}3 & j{=}4 \\ \hline
k{=}1 & 12 & 13 & 14 & 15 \\
k{=}2 & {\color{darkgreen} 12} & 23 & 24 & 25 \\
k{=}3 & {\color{darkgreen} 13} & {\color{darkgreen} 23} & 34 & 35 \\
k{=}4 & {\color{darkgreen} 14} & {\color{darkgreen} 24} & {\color{darkgreen} 34} & 45 \\
k{=}5 & {\color{darkgreen} 15} & {\color{darkgreen} 25} & {\color{darkgreen} 35} & {\color{darkgreen} 45}
\end{array}
\]
\end{example}

We now introduce $\boldsymbol{n}$\textbf{-simplex} equation, which turns out to be interesting because of its connection with the $\boldsymbol{(2n+1)}$\textbf{-gon}. All the content of the previous Subsection~\ref{ss:2n-simplex} and Proposition~\ref{p:A} remain perfectly valid if we change everywhere $2n$ to just~$n$. To distinguish analogues of sets~\eqref{Aq} from the ``$2n$''~case, we add a bar to letters $\mathcal A$ and~$a$ in the ``$n$''~case:
\begin{align}
\bar{\mathcal A}_k & = \{\bar a_{k,1},\ldots , \bar a_{k,n} \} , \label{Abar} \\[1ex]
    \bar a_{k,j} & = \left\{
    \begin{array}{ccc}
        \dfrac{(2n-k)(k-1)}{2}+j & \text{ for} & j\geq k , \\[1.5ex]
        \bar a_{j,k-1} & \text{ for} & j<k .
    \end{array} \right.
\end{align}

\begin{proposition}\label{p:B}
Positions where matrices act nontrivially in the $(2n+1)$-gon equation are given by
 \begin{align}
   & \mathcal{B}_{2k-1}=\bar{\mathcal A}_k, \label{B-odd} \\
   & \mathcal{B}_{2k} =\bar{\mathcal A}_k+b_k , \qquad
    \text{with} \quad b_k=\left(\beta _{k,1},\ldots,\beta _{k,n}\right),\qquad
     \beta _{k,j}=\left\{
        \begin{array}{cc}
             0, & j\geq k, \\
             1, & j<k .
        \end{array}
   \right. \label{B-even}
 \end{align}
\end{proposition}

\begin{proof}
It follows from the analysis made in the proof of Proposition~\ref{p:ind-rule} that $\mathcal B_{2k-1}$ can be described in terms of a \emph{fixed}---not changing with the action of any~$A^{(2k-1)}$---sequence of double indices. Namely, take, for each pair, the first member from the corresponding pair in~\eqref{gon}, while the second from the corresponding pair in~\eqref{rgon}. We obtain the following sequence involving \emph{only odd} numbers (where we highlight a ``typical subsequence'' like we did earlier):
\begin{equation}\label{2l-1}
 \begin{aligned}
  (1,3), & (1,5), \ldots, (1,2n+1), \ldots, \\
  & \boldsymbol{(2l-1, 2l+1), (2l-1,2l+3), \ldots, (2l-1, 2n+1)}, \\
  & \qquad \ldots, (2n-1, 2n+1),
 \end{aligned}
\end{equation}
and $\mathcal B_{2k-1}$ consists, as can be readily seen, exactly of such positions where $2k-1$ is a member of the corresponding pair.

Compare now~\eqref{2l-1} with the analogue of~\eqref{sim} for the $n$-simplex:
\begin{equation}\label{l}
 \begin{aligned}
  (1,2), & (1,3), \ldots, (1,n+1), \ldots, \\
  & \boldsymbol{(l,l+1), (l,l+2), \ldots, (l, n+1)}, \ldots, (n, n+1).
 \end{aligned}
\end{equation}
Equality~\eqref{B-odd} is now clear from the fact that for any number~$m$ entering in~\eqref{l}, there is~$(2m-1)$ at the corresponding place in~\eqref{2l-1}.

Similarly, $\mathcal B_{2k}$ can be described as consisting of such positions where $2k$ is a member of the pair in the following fixed---not changing with the action of any~$A^{(2k)}$---sequence: take, for each pair, the first member from the corresponding pair in~\eqref{rgon}, while the second from the corresponding pair in~\eqref{gon} (so, we write a pair of coinciding numbers if these two coincide). We obtain the following sequence involving \emph{only even} numbers (we highlight again a ``typical subsequence''):
\begin{equation}\label{2l}
 \begin{aligned}
  (2,2), & (2,4), \ldots, (2,2n), \ldots, \\
  & \boldsymbol{(2l,2l), (2l,2l+2), \ldots, (2l, 2n)}, \ldots, (2n, 2n),
 \end{aligned}
\end{equation}
Comparing \eqref{2l} with~\eqref{l}, one arrives, after a simple analysis, at equality~\eqref{B-even}.
\end{proof}

\subsection{Three-color decomposition}

As we have already said, each matrix~$R^{(q)}$ from the $2n$-simplex equation~\eqref{2n-simplex} acts nontrivially on its own $2n$-row, see~\eqref{row-Rq}. Consider now one special type of $R$-matrices, for which the \emph{even} entries of the output row (r.h.s.\ of~\eqref{row-Rq}) depend \emph{only on odd} entries of the input row (in the l.h.s.\ of~\eqref{row-Rq}), and vice versa. That is, we can represent in this case the dependence~\eqref{row-Rq} as
\begin{equation}\label{AB}
 \begin{aligned}
   \begin{pmatrix} u_{p_1,q} & u_{p_3,q} & \dots & u_{p_{2n-1},q} \end{pmatrix} A^{(q)} = \begin{pmatrix} v_{p_2,q} & v_{p_4,q} & \dots & v_{p_{2n},q} \end{pmatrix}, \\
   \begin{pmatrix} u_{p_2,q} & u_{p_4,q} & \dots & u_{p_{2n},q} \end{pmatrix} B^{(q)} = \begin{pmatrix} v_{p_1,q} & v_{p_3,q} & \dots & v_{p_{2n-1},q} \end{pmatrix} ,
 \end{aligned}
\end{equation}
where $A^{(q)}$ and~$B^{(q)}$ are the corresponding submatrices of~$R^{(q)}$ (and all entries of~$R^{(q)}$ outside $A^{(q)}$ and~$B^{(q)}$ are zero). Of course we called them $A^{(q)}$ and~$B^{(q)}$ intentionally, because the indexation in~\eqref{AB} coincides exactly with \eqref{gon-inds} and~\eqref{g-i-B}.

Such matrices can be said to admit a \emph{factorization} in $A^{(q)}$, $B^{(q)}$, and permutation matrices~$P_{i,j}$, because \eqref{AB} is equivalent to
\begin{equation}\label{RABP}
 \begin{aligned}
    R^{(q)} & = A_{1,3,\ldots ,2n-1}^{(q)}\, B_{2,4,\ldots ,2n}^{(q)}\, P_{1,2}P_{3,4} \cdots  P_{2n-1,2n} \\
    & = A_{1,3,\ldots ,2n-1}^{(q)}\, P_{1,2}P_{3,4} \cdots  P_{2n-1,2n} B_{1,3,\ldots ,2n-1}^{(q)} \, , \qquad \quad
    P=\left(\begin{array}{cc}0 & 1 \\1 & 0 \end{array}\right).
 \end{aligned}
\end{equation}
Subscripts in~\eqref{RABP} correspond to rows and columns of~$R^{(q)}$ viewed as a $(2n\times 2n)$-matrix.

We now assign colors---blue, red and green---to each entry in the initial $n(2n+\nobreak 1)$-row (on which both l.h.s.\ and r.h.s.\ of~\eqref{2n-simplex} act), as well as in all the middle rows and the final row, according to the following rules. In the initial row, the blue color is assigned to entries in positions, or simply ``to positions'',~\eqref{gon}; the red color is assigned to positions~\eqref{rgon}; and the green color to all the remaining positions. Then, taking either l.h.s.\ or r.h.s., we propagate the colors through each~$R^{(q)}$ in such way that the entries in the r.h.s.\ of~\eqref{AB} acquire the same color as those in the corresponding l.h.s.\ of~\eqref{AB}. It is important to note that this process goes ahead without obstacles or contradictions, and this is because the ``blue sector'' occupies exactly the places indexed in the same way as we did for the $(2n+1)$-gon equation~\eqref{(2n+1)-gon}, while the ``red sector''---indexed in the same way as for~\eqref{gon-inv}; the rest of the places remain for the green sector.

We now sum it up as the following proposition.

\begin{proposition}\label{p:d}
For $R$-matrices admitting factorization of type~\eqref{RABP}, the $2n$-simplex equation can be decomposed into three independent parts: ``blue'', ``red'' and ``green'', the blue part being the $(2n+1)$-gon equation~\eqref{(2n+1)-gon}, while the red part---the inverse $(2n+1)$-gon~\eqref{gon-inv} equation.   \qed
\end{proposition}

Thus, blue part involves only matrices~$A^{(q)}$, red part---only~$B^{(q)}$, and the green part, as one can see for instance from Example~\ref{x:d} below, is a mixed equation for both $A^{(q)}$ and $B^{(q)}$.

Some important properties of the green sector deserve a separate proposition.

\begin{proposition}\label{p:g}
 \hspace{0em} 
 \begin{itemize}\itemsep 0pt
  \item[$\mathrm{(i)}$] The initial and final positions in the green sector are the same, namely having both indices either odd or even:
   \begin{equation}\label{gif}
    \begin{aligned}
     & (2j,2k), \quad 1\le j<k \le n \\
     & \text{ or \ } (2j+1,2k+1), \quad 0\le j<k \le n.
    \end{aligned}
   \end{equation}
  \item[$\mathrm{(ii)}$] The inner (not initial or final) positions in the green sector all have one index odd and the other even, that is, belong to one of the sequences \eqref{gon} and~\eqref{rgon}. Moreover, any position in \eqref{gon} and~\eqref{rgon} is met exactly one time in the green sector.
 \end{itemize}
\end{proposition}

\begin{proof}
 \begin{itemize}\itemsep 0pt
  \item[$\mathrm{(i)}$] This follows from the fact that positions \eqref{gon} are initial and positions \eqref{rgon} are final for the blue sector, while positions \eqref{rgon} are initial and positions \eqref{gon} are final for the red sector; what remains is ``green''.
  \item[$\mathrm{(ii)}$]
   \begin{itemize}\itemsep 0pt
    \item[$\mathrm{(a)}$] First, we note that there are exactly two $R$-matrices changing the entry at position~$(j,k)$, namely $R^{(j)}$ and~$R^{(k)}$. So, there is one ``initial'' row entry at a given position, one ``inner'', and one ``final''.
    \item[$\mathrm{(b)}$] It follows that, for a chosen~$R^{(q)}$ and a position~$(q,l)\; \bigl({} = (l,q)\bigr)$, either the input row entry~$u_{q,l}$ is initial or the output row entry~$v_{q,l}$ is final. Indeed, the first case happens if $R^{(q)}$ acts before~$R^{(l)}$, while the second case happens if $R^{(q)}$ acts after~$R^{(l)}$.
    \item[$\mathrm{(c)}$] It follows then from~$\mathrm{(b)}$ that there is no matrix~$R^{(q)}$ belonging wholly to a single color sector. Indeed, it can be easily seen from~\eqref{gif} that, for a given~$q$, some of positions~$(q,l)\; \bigl({} = (l,q)\bigr)$ are green if considered as either initial or final, while some other are not.
    \item[$\mathrm{(d)}$] Thus, the color at each given position is changed with the action of any $R$-matrix. This together with the preceding analysis shows that there are exactly four possibilities for the initial--inner--final colors: {\color{blue}blue}--{\color{green}green}--{\color{red}red}, {\color{red}red}--{\color{green}green}--{\color{blue}blue}, {\color{green}green}--{\color{blue}blue}--{\color{green}green} and {\color{green}green}--{\color{red}red}--{\color{green}green}. Hence, item~$\mathrm{(ii)}$ follows.
   \end{itemize}
 \end{itemize}
\end{proof}

It follows from Proposition~\ref{p:g}~$\mathrm{(i)}$ that any side of~\eqref{2n-simplex} is a \emph{direct sum} of two submatrices, acting one in the green sector, while the other in the blue and red sectors. Moreover, the latter submatrix, due to the fact that $B^{(q)} = \left( A^{(q)} \right)^{-1}$ in our construction, has the obvious block structure
\begin{equation}\label{br-block}
\begin{pmatrix} 0 & {\color{blue}K^{-1}} \\ {\color{red}K} & 0 \end{pmatrix}.
\end{equation}

\begin{example}\label{x:d}
For $n=2$, we have a 4-simplex equation, a heptagon (7-gon) equation, and the following coloring of~\eqref{sim} (where we write simply ``{\color{blue}12}'' instead of~``{\color{blue}(1,2)}'', etc.):
\begin{equation*}
    ({\color{blue}12},{\color{green}13},{\color{blue}14},{\color{green}15},{\color{red}23},
    {\color{green}24},{\color{red}25},{\color{blue}34},{\color{green}35},{\color{red}45}).
\end{equation*}
The 4-simplex equation reads
\begin{equation}\label{4-simplex}
    R_{1234}^{(1)}R_{1567}^{(2)}R_{2589}^{(3)}R_{3680}^{(4)}R_{4790}^{(5)}
    =R_{4790}^{(5)}R_{3680}^{(4)}R_{2589}^{(3)}R_{1567}^{(2)}R_{1234}^{(1)},
\end{equation}
and can be re-written, using~\eqref{RABP}, as
\begin{align}
    {\color{blue}(A_{13}^{(1)}A_{18}^{(3)}A_{38}^{(5)})} &
    {\color{green}(B_{24}^{(1)}A_{26}^{(2)}B_{29}^{(3)}A_{49}^{(4)}B_{69}^{(5)})}
    {\color{red}(B_{57}^{(2)}B_{70}^{(4)})}  \nonumber \\
    {}={} & {\color{blue}(A_{38}^{(4)}A_{13}^{(2)})}
    {\color{green}(A_{49}^{(5)}B_{69}^{(4)}A_{29}^{(3)}B_{24}^{(2)}A_{26}^{(1)})}
    {\color{red}(B_{70}^{(5)}B_{50}^{(3)}B_{57}^{(1)})} , \label{4s-decomp}
\end{align}
where the equality holds also for each \emph{separate} color.

The left-hand side of either \eqref{4-simplex} or~\eqref{4s-decomp} can be visualized as in Figure~\ref{fig:4-simplex}.
\end{example}

\begin{figure}[t]
 \begin{center}
  \begin{tikzpicture}

\draw[blue, thick, -> ] (0,0) -- (0,0.9) ;
\draw[green, thick, -> ] (1,0) -- (1,0.9) ;
\draw[blue, thick, -> ] (2,0) -- (2,0.9) ;
\draw[green, thick, -> ] (3,0) -- (3,0.9) ;
\draw[red, thick, -> ] (4,0) -- (4,1.9) ;
\draw[green, thick, -> ] (5,0) -- (5,1.9) ;
\draw[red, thick, -> ] (6,0) -- (6,1.9) ;
\draw[blue, thick, -> ] (7,0) -- (7,2.9) ;
\draw[green, thick, -> ] (8,0) -- (8,2.9) ;
\draw[red, thick, -> ] (9,0) -- (9,3.9) ;

\draw[green, thick, -> ] (0,1) -- (0,1.9) ;
\draw[blue, thick, -> ] (1,1) -- (1,2.9) ;
\draw[green, thick, -> ] (2,1) -- (2,3.9) ;
\draw[blue, thick, -> ] (3,1) -- (3,4.9) ;
\draw[green, thick, -> ] (4,2) -- (4,2.9) ;
\draw[red, thick, -> ] (5,2) -- (5,3.9) ;
\draw[green, thick, -> ] (6,2) -- (6,4.9) ;
\draw[green, thick, -> ] (7,3) -- (7,3.9) ;
\draw[blue, thick, -> ] (8,3) -- (8,4.9) ;
\draw[green, thick, -> ] (9,4) -- (9,4.9) ;

\draw[red, thick, -> ] (0,2) -- (0,6) ;
\draw[green, thick, -> ] (1,3) -- (1,6) ;
\draw[red, thick, -> ] (2,4) -- (2,6) ;
\draw[green, thick, -> ] (3,5) -- (3,6) ;
\draw[blue, thick, -> ] (4,3) -- (4,6) ;
\draw[green, thick, -> ] (5,4) -- (5,6) ;
\draw[blue, thick, -> ] (6,5) -- (6,6) ;
\draw[red, thick, -> ] (7,4) -- (7,6) ;
\draw[green, thick, -> ] (8,5) -- (8,6) ;
\draw[blue, thick, -> ] (9,5) -- (9,6) ;

\draw (0,1) node[black, anchor=east] {$R_{1234}^{(1)}={{\color{blue}A_{13}^{(1)}}{\color{green}B_{24}^{(1)}}P_{12}P_{34}} \hspace{ 1em }$} ;
\filldraw[gray] (0,0.95) rectangle (3,1.05) ;
\filldraw[gray] (0,1) circle (0.1);
\filldraw[gray] (1,1) circle (0.1);
\filldraw[gray] (2,1) circle (0.1);
\filldraw[gray] (3,1) circle (0.1);

\draw (0,2) node[black, anchor=east] {$R_{1567}^{(2)}={{\color{green}A_{16}^{(2)}}{\color{red}B_{57}^{(2)}}P_{15}P_{67}} \hspace{ 1em }$} ;
\filldraw[gray] (0,1.95) rectangle (6,2.05) ;
\filldraw[gray] (0,2) circle (0.1);
\filldraw[gray] (4,2) circle (0.1);
\filldraw[gray] (5,2) circle (0.1);
\filldraw[gray] (6,2) circle (0.1);

\draw (0,3) node[black, anchor=east] {$R_{2589}^{(3)}={{\color{blue}A_{28}^{(3)}}{\color{green}B_{59}^{(3)}}P_{25}P_{89}} \hspace{ 1em }$} ;
\filldraw[gray] (1,2.95) rectangle (8,3.05) ;
\filldraw[gray] (1,3) circle (0.1);
\filldraw[gray] (4,3) circle (0.1);
\filldraw[gray] (7,3) circle (0.1);
\filldraw[gray] (8,3) circle (0.1);

\draw (0,4) node[black, anchor=east] {$R_{3680}^{(4)}={{\color{green}A_{38}^{(4)}}{\color{red}B_{60}^{(4)}}P_{36}P_{80}} \hspace{ 1em }$} ;
\filldraw[gray] (2,3.95) rectangle (9,4.05) ;
\filldraw[gray] (2,4) circle (0.1);
\filldraw[gray] (5,4) circle (0.1);
\filldraw[gray] (7,4) circle (0.1);
\filldraw[gray] (9,4) circle (0.1);

\draw (0,5) node[black, anchor=east] {$R_{4790}^{(5)}={{\color{blue}A_{49}^{(5)}}{\color{green}B_{70}^{(5)}}P_{47}P_{90}} \hspace{ 1em }$} ;
\filldraw[gray] (3,4.95) rectangle (9,5.05) ;
\filldraw[gray] (3,5) circle (0.1);
\filldraw[gray] (6,5) circle (0.1);
\filldraw[gray] (8,5) circle (0.1);
\filldraw[gray] (9,5) circle (0.1);

\draw (0,-0.3) node[blue, anchor=center] {$12$} ;
\draw (1,-0.3) node[green, anchor=center] {$13$} ;
\draw (2,-0.3) node[blue, anchor=center] {$14$} ;
\draw (3,-0.3) node[green, anchor=center] {$15$} ;
\draw (4,-0.3) node[red, anchor=center] {$23$} ;
\draw (5,-0.3) node[green, anchor=center] {$24$} ;
\draw (6,-0.3) node[red, anchor=center] {$25$} ;
\draw (7,-0.3) node[blue, anchor=center] {$34$} ;
\draw (8,-0.3) node[green, anchor=center] {$35$} ;
\draw (9,-0.3) node[red, anchor=center] {$45$} ;

\draw (0,6.3) node[red, anchor=center] {$12$} ;
\draw (1,6.3) node[green, anchor=center] {$13$} ;
\draw (2,6.3) node[red, anchor=center] {$14$} ;
\draw (3,6.3) node[green, anchor=center] {$15$} ;
\draw (4,6.3) node[blue, anchor=center] {$23$} ;
\draw (5,6.3) node[green, anchor=center] {$24$} ;
\draw (6,6.3) node[blue, anchor=center] {$25$} ;
\draw (7,6.3) node[red, anchor=center] {$34$} ;
\draw (8,6.3) node[green, anchor=center] {$35$} ;
\draw (9,6.3) node[blue, anchor=center] {$45$} ;

\draw (0,-0.7) node[anchor=center] {$\scriptstyle 1$} ;
\draw (1,-0.7) node[anchor=center] {$\scriptstyle 2$} ;
\draw (2,-0.7) node[anchor=center] {$\scriptstyle 3$} ;
\draw (3,-0.7) node[anchor=center] {$\scriptstyle 4$} ;
\draw (4,-0.7) node[anchor=center] {$\scriptstyle 5$} ;
\draw (5,-0.7) node[anchor=center] {$\scriptstyle 6$} ;
\draw (6,-0.7) node[anchor=center] {$\scriptstyle 7$} ;
\draw (7,-0.7) node[anchor=center] {$\scriptstyle 8$} ;
\draw (8,-0.7) node[anchor=center] {$\scriptstyle 9$} ;
\draw (9,-0.7) node[anchor=center] {$\scriptstyle 0$} ;

\end{tikzpicture}
 \end{center}
\caption{The left-hand side of~\eqref{4-simplex} corresponds to going from the bottom to the top of this picture}
\label{fig:4-simplex}
\end{figure}
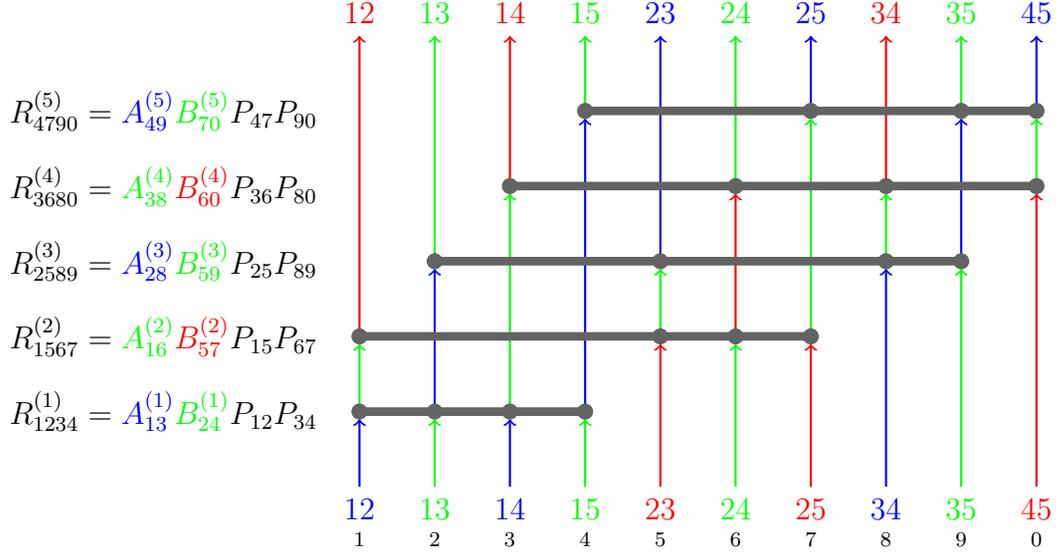

\subsection[Reduction to $(2n-1)$-simplex]{Reduction to $\boldsymbol{(2n-1)}$-simplex}\label{ss:reduction}

If we have a solution of direct-sum $2n$-simplex equation~\eqref{2n-simplex}, we can, under some mild conditions, apply to it a \emph{reduction} procedure and obtain a solution of the ``lower'' equation, namely $(2n-1)$-simplex. Here we describe such a procedure---an analogue of taking ``partial trace'' in the quantum case---and, specifically, what it gives for $R$-operators~\eqref{RABP}.

Each $R$-matrix in~\eqref{2n-simplex} imposes $2n$ linear conditions on $4n$~variables---$2n$ input and $2n$ output. In this Subsection, we would like to denote its input and output rows as follows:
\begin{equation}\label{Rio}
\begin{pmatrix} u_1^{(q)} & u_2^{(q)} & \ldots & u_{2n}^{(q)} \end{pmatrix} R^{(q)} = \begin{pmatrix} v_1^{(q)} & v_2^{(q)} & \ldots & v_{2n}^{(q)} \end{pmatrix}.
\end{equation}
We take then a number $\lambda \in F$ and impose the following additional condition:
\begin{equation}\label{vu}
u_{2n}^{(q)} = \lambda v_{2n}^{(q)},
\end{equation}
and do it for all $R$-matrices except the last, that is, for $R_{\mathcal A_1}^{(1)}$, \ldots, $R_{\mathcal A_{2n}}^{(2n)}$. We assume that, for each of these $R$-matrices, \eqref{Rio} together with~\eqref{vu} determine also~$u_{2n}^{(q)}$ from the rest of input variables (which is usually the case). Hence, in particular, variables $v_1^{(q)}, v_2^{(q)}, \ldots, v_{2n-1}^{(q)}$ are determined given $u_1^{(q)}, u_2^{(q)}, \ldots, u_{2n-1}^{(q)}$, and we denote~$Z^{(q)}$ the matrix giving these dependencies:
\begin{equation}\label{Z}
\begin{pmatrix} u_1^{(q)} & u_2^{(q)} & \ldots & u_{2n-1}^{(q)}  \end{pmatrix} Z^{(q)} = \begin{pmatrix} v_1^{(q)} & v_2^{(q)} & \ldots & v_{2n-1}^{(q)}  \end{pmatrix}.
\end{equation}

\begin{proposition}\label{p:Z}
Matrices $Z^{(q)}$ satisfy the $(2n-1)$-simplex equation
\begin{equation}\label{(2n-1)-simplex}
Z^{(1)} Z^{(2)} \dots Z^{(2n)} = Z^{(2n)} \dots Z^{(2)} Z^{(1)}.
\end{equation}
\end{proposition}

In~\eqref{(2n-1)-simplex}, we did not write out the positions where matrices~$Z^{(q)}$ act; these are defined in complete analogy with Subsection~\ref{ss:2n-simplex}.

\begin{proof}
First, write a special input row for the whole $2n$-simplex~\eqref{2n-simplex}, choosing last entries~$u_{2n}^{(q)}$ obeying~\eqref{vu} for all $R$-matrices in the l.h.s.\ except~$R^{(2n+1)}$: first for~$R^{(1)}$, then for~$R^{(2)}$, \ldots, finally for~$R^{(2n)}$. No conditions are imposed on the other input entries.

As $\lambda$ is the same for all the mentioned~$R^{(q)}$, and $R^{(2n+1)}$ acts exactly on the last positions of all other matrices, the input rows for~$R^{(2n+1)}$ in the l.h.s.\ and r.h.s.\ of~\eqref{2n-simplex} are proportional: in the r.h.s., they are $\lambda$~times those in the l.h.s. Hence, the same holds also for the output rows, and we can describe what happens at such a position under action of two matrices as follows (denoting $w_{2n}^{(q)}$ the final value at this position):
\begin{align}
\text{l.h.s.:} & \quad u_{2n}^{(q)} = \lambda v_{2n}^{(q)} \xmapsto{\;\; R^{(q)}\;\;} v_{2n}^{(q)} \xmapsto{R^{(2n+1)}} w_{2n}^{(q)}, \label{lw} \\[1ex]
\text{r.h.s.:} & \quad u_{2n}^{(q)} = \lambda v_{2n}^{(q)} \xmapsto{R^{(2n+1)}} \lambda w_{2n}^{(q)} \xmapsto{\;\;R^{(q)}\;\;} w_{2n}^{(q)}. \label{rw}
\end{align}
Indeed, the values around the first arrow in~\eqref{rw} must be proportional to those around the second arrow in~\eqref{lw}. But then we see from the second arrow in~\eqref{rw} that the same condition as~\eqref{vu} holds for~$R^{(q)}$ also in the r.h.s., hence, its reduction~$Z^{(q)}$ is the same matrix in both sides, and \eqref{(2n-1)-simplex} holds.
\end{proof}

Explicit expressions for matrices~$Z^{(q)}$ in the case of~$R^{(q)}$ as in~\eqref{RABP} are given by
\begin{equation}\label{Zex}
 \begin{aligned}
Z^{(q)}=A_{1,3,\ldots ,2n-1}^{(q)}\, P_{12}P_{34}\cdots  P_{2n-3,2n-2}\, \Lambda _{2n-1}\, B_{1,3,\ldots ,2n-1}^{(q)},\\
 \Lambda _{2n-1}=\left(\begin{array}{cc} \mathbb 1_{2n-2} & 0 \\ 0 & \lambda  \\ \end{array} \right),
 \end{aligned}
\end{equation}
as a small exercise shows. Formula~\eqref{Zex} can be treated as a generalization of (the direct-sum analogue of) the tetrahedron equation solution in~\cite[(1.8)]{KashS}.

Note that $\lambda = 0$ can also be taken in~\eqref{vu}, and the obtained solutions in this case are
\[
Z^{(q)}=S^{\mathrm T} R^{(q)} S,\qquad S=\left(\begin{array}{c} \mathbb 1_{2n-1} \\ 0 \\ \end{array} \right),\qquad q=1,2,\ldots  ,2n-1,
\]
that is, $Z^{(q)}$ is simply the submatrix of~$R^{(q)}$ obtained by deleting its last row and column.

Finally, we note that the described kind of reduction can be applied not one but several times: from $(2n-1)$-simplex to $(2n-2)$-simplex and so on. The solutions of simplex equations thus obtained deserve to be the subject of a separate research.

\section{Construction of solutions}\label{s:sols}

\subsection[Multivectors $\phi $ and $\psi $ on which the construction is based]{Multivectors $\boldsymbol{\phi }$ and $\boldsymbol{\psi }$ on which the construction is based}\label{ss:mv}

Let $F$ be a field, $F^{2n+1}$ the space of $(2n+1)$-rows, $\mathsf e_1,\ldots  ,\mathsf e_{2n+1}$ the standard basis in~$F^{2n+1}$, and $\mathsf e^1,\ldots  ,\mathsf e^{2n+1}$ the dual basis in the dual space. Let
\begin{equation}\label{M}
\mathcal{M} = \begin{pmatrix}
 \alpha _{1,1} & \alpha _{1,2} & \cdots  & \alpha _{1,2n+1} \\
 \alpha _{2,1} & \alpha _{2,2} & \cdots  & \alpha _{2,2n+1} \\
 \vdots  & \vdots  & \ddots & \vdots  \\
 \alpha _{n+1,1} & \alpha _{n+1,2} & \cdots  & \alpha _{n+1,2n+1} \\
\end{pmatrix} 
\end{equation}
be a matrix defining an element~$L$ of the Grassmannian $\Gr(n+1,2n+1)$, that is, a matrix of the full rank~$n+1$ whose rows, which can also be written as
\[
v_i=\alpha _{i,1}\mathsf e_1+\dots+\alpha _{i,2n+1}\mathsf e_{2n+1},
\]
span an $(n+1)$-dimen\-sional plane~$L\subset F^{2n+1}$.

Below, we will be working with \emph{multivectors}---elements of the \emph{exterior algebra} $\bigwedge F^{2n+1}$ over~$F^{2n+1}$. First, we introduce an $(n+1)$-vector---the exterior product of all~$v_i$:
\[
w =v_1\wedge v_2\wedge \cdots \wedge v_{n+1}=\sum _{k_1<\cdots <k_{n+1}} p_{k_1,\ldots ,k_{n+1}}\mathsf e_{k_1}
\wedge \dots \wedge \mathsf e_{k_{n+1}},
\]
where $p_{k_1,\ldots ,k_{n+1}}$ are determinants made of $k_1$-th, \ldots, $k_{n+1}$-th columns of~$\mathcal{M}$, called also \emph{Pl{\" u}cker coordinates} of the Grassmannian $\Gr(n+1,2n+1)$. 

We will also need the \emph{convolution} operation~$\,\oplrcorner\,$, see, for instance,~\cite[page~42]{Shaf}. In our context, this will be essentially the same as the left \emph{Grassmann derivative}~\cite[Eq.~(2.1.2)]{Ber}: for instance, $\phi ^{i,j}$ below in formula~\eqref{phi} can be written also as $\frac{\partial}{\partial \mathsf e_j} \frac{\partial}{\partial \mathsf e_i} w$.

\begin{assumption}\label{a:p-nonzero}
Below, we assume that \emph{all Pl{\" u}cker coordinates are nonzero.}
\end{assumption}

Under Assumption~\ref{a:p-nonzero}, field~$F$ must not be too small, as the following example shows. A deeper analysis using the ``irrelevance of algebraic inequalities'' principle shows that Assumption~\ref{a:p-nonzero} can be partially relaxed, because what we ultimately need is that there must be no division by zero in formulas \eqref{A-def} and~\eqref{B-def} below. This may be important when working with finite fields; we leave these questions for future research.

\begin{example}\label{x:GF4}
For $n=2$, the smallest field allowing to construct matrix~$\mathcal M$~\eqref{M} obeying Assumption~\ref{a:p-nonzero} is the Galois field~$\mathbb F_4$ of four elements. For instance,
\[
\mathcal M = \begin{pmatrix} 1 & 0 & 0 & 1 & 1 \\ 0 & 1 & 0 & 1 & \zeta \\ 0 & 0 & 1 & 1 & \zeta^2 \end{pmatrix},
\]
where $\zeta \in \mathbb F_4$ is not equal to either $0$ or~$1$.
\end{example}

We define $(n-1)$-vectors
\begin{equation}\label{phi}
\phi ^{i,j} = \left(\mathsf e^j\wedge \mathsf e^i\right) \oplrcorner w =\sum _{k_1<\cdots <k_{n-1}} p_{i,j,k_1,\ldots ,k_{n-1}}\mathsf e_{k_1}
\wedge \cdots \wedge \mathsf e_{k_{n-1}},
\end{equation}
and $(n+3)$-vectors
\begin{equation}\label{psi}
\psi _{i,j}=\mathsf e_i\wedge \mathsf e_j\wedge w =\sum _{k_1<\cdots <k_{n+1}} p_{k_1,\ldots ,k_{n+1}}\mathsf e_i\wedge \mathsf e_j\wedge \mathsf e_{k_1}\wedge \cdots \wedge \mathsf e_{k_{n+1}}.
\end{equation}
Both objects $\phi ^{i,j}$ and $\psi _{i,j}$ are antisymmetric in their indices.

\begin{example}
For $n=3$ we have a 2-vector
\begin{align*}
\phi ^{1,2}=\sum _{1\leq k<l\leq 7} p_{1,2,k,l}\mathsf e_k\wedge \mathsf e_l=-p_{1,2,3,4} \mathsf e_3\wedge \mathsf e_4-p_{1,2,3,5} \mathsf e_3\wedge \mathsf e_5-p_{1,2,3,6} \mathsf e_3\wedge \mathsf e_6\\
{}-p_{1,2,3,7} \mathsf e_3\wedge \mathsf e_7-p_{1,2,4,5} \mathsf e_4\wedge \mathsf e_5-p_{1,2,4,6} \mathsf e_4\wedge \mathsf e_6-p_{1,2,4,7} \mathsf e_4\wedge \mathsf e_7\\
{}-p_{1,2,5,6} \mathsf e_5\wedge \mathsf e_6-p_{1,2,5,7} \mathsf e_5\wedge \mathsf e_7-p_{1,2,6,7} \mathsf e_6\wedge \mathsf e_7
\end{align*}
and a 6-vector
\begin{align*}
\psi _{2,3}=p_{1,4,5,6} \mathsf e_1\wedge \mathsf e_2\wedge \mathsf e_3\wedge \mathsf e_4\wedge \mathsf e_5\wedge \mathsf e_6+p_{1,4,5,7} \mathsf e_1\wedge \mathsf e_2\wedge \mathsf e_3\wedge \mathsf e_4\wedge \mathsf e_5\wedge \mathsf e_7\\
{}+p_{1,4,6,7} \mathsf e_1\wedge \mathsf e_2\wedge \mathsf e_3\wedge \mathsf e_4\wedge \mathsf e_6\wedge \mathsf e_7+p_{1,5,6,7} \mathsf e_1\wedge \mathsf e_2\wedge \mathsf e_3\wedge \mathsf e_5\wedge \mathsf e_6\wedge \mathsf e_7 \\ 
{}+p_{4,5,6,7} \mathsf e_2\wedge \mathsf e_3\wedge \mathsf e_4\wedge \mathsf e_5\wedge \mathsf e_6\wedge \mathsf e_7 .
\end{align*}
\end{example}

\subsection{Dimensions of linear spaces spanned by multivectors}\label{ss:dims}

\begin{proposition}\label{p:le-n}
If we fix index~$j$ in $\phi^{i,j}$, while letting $i$ take any $n$ different values $i\ne j$, then the $n$ resulting multivectors are linearly independent. 
\end{proposition}

\begin{proof}
Denote the chosen values of~$i$ as $i_1 < \dots < i_n$. Suppose there is a vanishing linear combination
\begin{equation}\label{lin-comb-phi}
\sum_{l=1}^{n} \lambda_l \phi^{{i_l},j} = 0
\end{equation}
of the mentioned multivectors. We are going to show that all its coefficients~$\lambda_l$ are zero.

Fix $l$, and consider, for all $i=i_1,\ldots, i_n$, the corresponding coefficients of $\mathsf e_{i_1} \wedge \dots \wedge \hat{\mathsf e}_{i_l} \wedge \dots \wedge \mathsf e_{i_n} $ in the r.h.s.\ of~\eqref{phi}. Clearly, they all vanish except the $l$-th one which is $\lambda_l p_{i_l,j, i_1, \dots, \hat{i}_l, \dots, i_n }$. Combined with Assumption~\ref{a:p-nonzero}, this immediately gives $\lambda_l=0$.
\end{proof}

\begin{proposition}\label{p:exactly-n}
The linear space spanned by multivectors $\phi^{i,j}$ with a fixed~$j$ is exactly $n$-dimen\-sional.
\end{proposition}

\begin{proof}
There are only $(n+1)$ linearly independent among expressions $\mathsf e^i \oplrcorner w$ for all~$i$. Further action of~$\mathsf e^j$, that is, taking $\left(\mathsf e^j \wedge \mathsf e^i \right) \oplrcorner w$ as in~\eqref{phi}, kills one of these, so, $n$~linearly independent expressions remain.
\end{proof}

\begin{proposition}\label{p:phi-plgn}
The multivectors~$\phi^{i,j}$ on which either l.h.s.\ or r.h.s.\ of the $(2n+ \nobreak 1)$-gon relation acts, are linearly independent, that is, span an $\dfrac{n(n+1)}{2}$-dimen\-sional linear space.
\end{proposition}

\begin{proof}
We have to show that
\[
\sum_{j=1}^n \sum_{k= j}^n \lambda _{2 j-1,2 k} \phi ^{2 j-1,2 k} = 0
\]
implies $\lambda _{2 j-1,2 k} = 0$ for all $j, k$ entering in the above double sum.

We note first that
\[
\mathsf e_1 \wedge  \mathsf e_3 \wedge  \dots  \wedge  \mathsf e_{2 n-5} \wedge  \mathsf e_{2 n-3}\, p_{ 1,3,\dots ,2 n-3,2 n-1,2 n}
\]
appears only in $\phi^{2 n-1,2 n}$ hence $\lambda_{2 n-1,2 n} = 0$. After that we note next that
\[
\mathsf e_1 \wedge  \mathsf e_3 \wedge  \dots  \wedge  \mathsf e_{2 n-5} \wedge  \mathsf e_{2 n-2}\, p_{1,3,\dots ,2 n-3,2 n-2,2 n}
\]
appears only in $\phi^{2 n-3,2 n}$ hence $\lambda_{2 n-3,2 n} = 0$. Continuing in this way suppose we have proved that $\lambda_{2 n-2 j-1,2 n} = 0$ for $j = 1, 2, \dots  , k - 1$,
then the term
\[
\mathsf e_1 \wedge  \mathsf e_3 \wedge  \dots  \wedge  \mathsf e_{2 n-2 k-1} \wedge  \mathsf e_{2 n-2 k} \wedge  \dots  \wedge  \mathsf e_{2 n-2}\, p_{1,3,\dots ,2 n-2 k-1,2 n-2 k,\dots ,2 n-2,2 n}
\]
appears only in $\phi^{2 n-2 k-1,2 n}$ hence $\lambda_{2 n-2 k-1,2 n} = 0$. By induction $\lambda _{2 k-1,2 n} = 0$ for $k = 1, 2, \dots , n$.

Next we look for the term
\[
\mathsf e_1 \wedge  \mathsf e_3 \wedge  \dots  \wedge  \mathsf e_{2 n-3} \wedge  \mathsf e_{2 n-2} \wedge  \mathsf e_{2 n}\, p_{1,3,\dots ,2 n-3,2 n-2,2 n}
\]
from which $\lambda_{2 n-3,2 n-2} = 0$ follows, etc.
\end{proof}

Now a proposition about $\phi$'s and~$\psi$'s in the \emph{green} sector of the $2n$-simplex relation. The green sector acts on those pairs $i,j$ where either both $i$ and~$j$ are even, or both $i$ and~$j$ are odd. It is enough for us to consider just the odd $i$ and~$j$ for~$\phi$'s, and just the even $i$ and~$j$ for~$\psi$'s.

\begin{proposition}\label{p:green}
 \hspace{0em} 
 \begin{itemize}\itemsep 0pt
  \item[$\mathrm{(i)}$] The $\dfrac{n(n+1)}{2}$ multivectors~$\phi_{i,j}$ for all odd $i$ and~$j$ are linearly independent.
  \item[$\mathrm{(ii)}$] The $\dfrac{n(n-1)}{2}$ multivectors~$\psi_{i,j}$ for all even $i$ and~$j$ are linearly independent.
 \end{itemize}
\end{proposition}

We will see when proving Theorem~\ref{th:simplex} that actually the space spanned by \emph{all}~$\phi_{i,j}$ in the green sector is \emph{exactly} $\frac{n(n+1)}{2}$-dimen\-sional, while the space spanned by \emph{all}~$\psi_{i,j}$ in the green sector is \emph{exactly} $\frac{n(n-1)}{2}$-dimen\-sional.

\begin{proof}
The proofs of $\mathrm{(i)}$ and~$\mathrm{(ii)}$ are almost identical, only with some obvious changes, so it is enough to write out here only the proof of~$\mathrm{(ii)}$. Suppose, like in the proofs of Propositions \ref{p:le-n} and~\ref{p:phi-plgn}, that there is a dependence
\[
\sum_{\substack{i,j\mathrm{\;even} \\ i<j }} \lambda_{i,j} \psi_{i,j} = 0,
\]
choose any $i=i_0$ and $j=j_0$, and consider the coefficients of
\[\mathsf e_{i_0} \wedge \mathsf e_{j_0} \wedge \mathsf e_1  \wedge \mathsf e_3 \wedge \dots \wedge \mathsf e_{2n+1},
\]
using now of course the r.h.s.\ of~\eqref{psi} for~$\psi_{i,j}$. The same argument as in the mentioned proofs shows that $\lambda_{i_0,j_0}=0$.
\end{proof}

\subsection[Matrices $A^{(q)}$ and $B^{(q)}$ and their relations to $\phi $ and $\psi $]{Matrices $\boldsymbol{A^{(q)}}$ and $\boldsymbol{B^{(q)}}$ and their relations to $\boldsymbol{\phi }$ and $\boldsymbol{\psi }$}

We are going to define matrices $A^{(q)}$ and~$B^{(q)}=(A^{(q)})^{-1}$ that will give us solutions to both $(2n+1)$-gon and $2n$-simplex equations. Technically, we prefer to give here first explicit expressions \eqref{A-def} and~\eqref{B-def} for their matrix entries, and then prove their key properties in Proposition~\ref{p:AB}. It must be noted, however, that, \emph{conceptually}, the reason for existence of $A^{(q)}$ and~$B^{(q)}$ satisfying \eqref{phi-A} and~\eqref{phi-B} below lies in the dimension count given in Propositions \ref{p:le-n} and~\ref{p:exactly-n}. That is, the mentioned propositions guarantee that formulas \eqref{phi-A} and~\eqref{phi-B} below can be used as correct \emph{definitions} of $A^{(q)}$ and~$B^{(q)}$ respectively (and as for Proposition~\ref{p:phi-plgn}, it will be used for \emph{proving} the $(2n+1)$-gon relation, see Theorem~\ref{th:polygon}).

For $q\in [2n+1]=\{1,2,\ldots  ,2n+1\}$ we denote
\begin{equation*}
L_q=[2n+1]\backslash \{q\}=\left\{a_1,a_2,\ldots  ,a_{2n}\right\}
\end{equation*}
(slightly changing notations with respect to~\eqref{Lq}; remember also Convention~\ref{c:o}) and define the following two families of $(n\times n)$-matrices:
\begin{align}
\left(A^{(q)}\right)_i^j=(-1)^i \, \frac{p_{a_{2j},\, a_1,a_3,\ldots ,\hat{a}_{2i-1},\ldots ,a_{2n-1},\, q}}{p_{a_1,a_3,\ldots ,a_{2n-1},\,q}}, \label{A-def} \\
\left(B^{(q)}\right)_i^j=(-1)^i \, \frac{p_{a_{2j-1},\, a_2,a_4,\ldots ,\hat{a}_{2i},\ldots ,a_{2n},\,q}}{p_{a_2,a_4,\ldots ,a_{2n},\,q}}, \label{B-def} \\
q=1,2,\ldots ,2n+1 . \nonumber
\end{align}
Here index~$i$ numbers the rows, while $j$---the columns of an $n\times n$ matrix.

\begin{proposition}\label{p:AB}
The following equalities hold:
\begin{align}
\sum _{i=1}^n \phi ^{a_{2i-1},q} \left(A^{(q)}\right)_i^j=-\phi ^{a_{2j},q} , \label{phi-A} \\
\sum_{j=1}^n \left(A^{(q)}\right)_i^j\psi _{a_{2j},q}=\psi _{a_{2i-1,q}} , \label{psi-A} \\
\sum _{i=1}^n \phi ^{a_{2i},q}\left(B^{(q)}\right)_i^j=-\phi ^{a_{2j-1},q}, \label{phi-B} \\
\sum _{j=1}^n \left(B^{(q)}\right)_i^j\psi_{a_{2j-1},q}=\psi _{a_{2i,q}}. \label{psi-B}
\end{align}
\end{proposition}

Note that $A^{(q)}$ and~$B^{(q)}$ act on the \emph{rows} of~$\phi $'s, but on the \emph{columns} of~$\psi $'s. We have thus extended such formulas as~\eqref{AB} from just scalar row entries to entries taking values in multidimensional linear spaces (and introduced similar column entries). Note also that \eqref{BA} follows from \eqref{phi-A} and~\eqref{phi-B}, because the~$\phi$'s in either side of these formulas span an $n$-dimensional space and thus \eqref{phi-A} and~\eqref{phi-B} unambiguously determine $A^{(q)}$ and~$B^{(q)}$, respectively.

\begin{proof}
Equality~\eqref{phi-A} follows from the Pl{\" u}cker relations
\[
\sum _{i=1}^n (-1)^ip_{a_{2i},q,b_1,\ldots ,b_{n-1}} p_{a_{2j-1},a_2,\ldots ,\hat{a}_{2i},\ldots ,a_{2n},q}+p_{a_{2j-1},q,b_1,\ldots
,b_{n-1}}p_{a_2,\ldots ,a_{2n},q}=0,
\]
while~\eqref{psi-A} again from the Pl{\" u}cker relations in the form
\[
\bigl(\left(\mathsf e^{a_1}\wedge \mathsf e^{a_3}\wedge \cdots \wedge \hat{\mathsf e}^{a_{2i-1}}\wedge \cdots \wedge \mathsf e^{a_{2n-1}}\wedge \mathsf e^q\right) \oplrcorner w \bigr)\wedge w = 0.
\]
Equalities~\eqref{phi-B} and~\eqref{psi-B} are proved similarly.
\end{proof}

Note that
\begin{align*}
\sum _{j=1}^n p_{a_{2j},a_1,\ldots ,\hat{a}_{2i-1},\ldots ,a_{2n-1},q}\mathsf e_{a_{2j}}-(-1)^ip_{a_1,a_3,\ldots ,a_{2n-1},q}\mathsf e_{a_{2i-1}}\\
=\left(\mathsf e^{a_1}\wedge \mathsf e^{a_3}\wedge \cdots \wedge \hat{\mathsf e}^{a_{2i-1}}\wedge \cdots \wedge \mathsf e^{a_{2n-1}}\wedge \mathsf e^q\right) \oplrcorner w\, .
\end{align*}

\subsection[$(2n+1)$-gon and $2n$-simplex]{$\boldsymbol{(2n+1)}$-gon and $\boldsymbol{2n}$-simplex}

We are going to prove that matrices~$A^{(q)}$ given by~\eqref{A-def} give solution to $(2n+1)$-gon equation, by using rows of multivectors~$\phi^{i,j}$. We would like to explain this in terms of representing the $(2n+1)$-gon relation \emph{diagrammatically}, like in Figure~\ref{fig:pentagon}. The input and output $\phi$'s are thus thought of as placed on the ``legs'' coming from one matrix~$A^{(q)}$ to another, or ``initial'' or ``final'' legs, corresponding to ``input'' and ``output'' row entries for the whole l.h.s.\ or r.h.s.

We write, in the l.h.s., $\phi^{i,j}$ at the leg going \emph{from} $A^{(i)}$ \emph{to}~$A^{(j)}$, while in the r.h.s., we write $\phi^{j,i}=-\phi^{i,j}$ at such leg. And in the case if the leg of~$A^{(i)}$ is ``initial'' or ``final'', the missing index~$j$ is taken from the matrix having the corresponding ``initial'' or ``final'' leg in the other side of the equation.

It is important that, due to the antisymmetry of~$\phi^{i,j}$ in $i$ and~$j$, this agrees with~\eqref{phi-A} (just change the order of indices of~$\phi$ in the r.h.s.\ of~\eqref{phi-A}, because the corresponding legs go out of~$A^{(q)}$, and the minus in~\eqref{phi-A} disappears, as desired).

For $B^{(q)}$, we write $\phi $'s in the very same way, and this again agrees with~\eqref{phi-B}.

We will also prove in this Subsection that our $A^{(q)}$ and~$B^{(q)}$ give a solution of \emph{$2n$-simplex} equation (via~\eqref{RABP}). In that proof, we will be using rows of~$\phi $'s in the same fashion as above, but also we will be using \emph{columns} of~$\psi $'s.

\begin{example}
The case yielding matrices~$A^{(q)}$ for \emph{pentagon} equation is given in Figure~\ref{fig:penta-phi}
 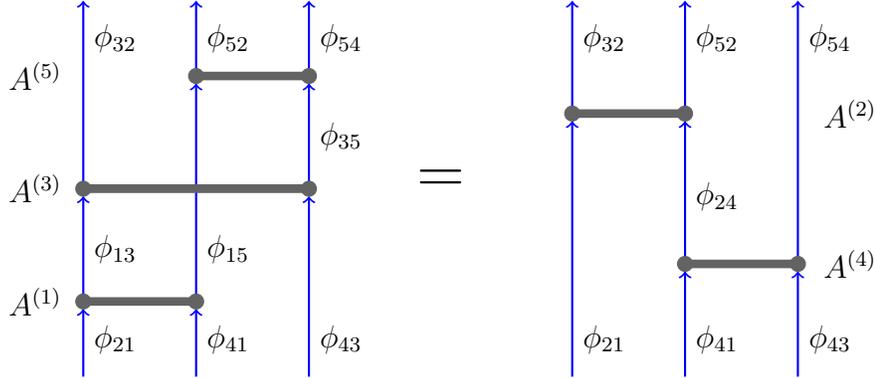
\begin{figure}
  \begin{center}
   \begin{tikzpicture}

\draw[blue, thick, -> ] (0,0) -- (0,0.9) ;
\draw[blue, thick, -> ] (1.5,0) -- (1.5,0.9) ;
\draw[blue, thick, -> ] (3,0) -- (3,2.4) ;

\draw[blue, thick, -> ] (0,1) -- (0,2.4) ;
\draw[blue, thick, -> ] (1.5,1) -- (1.5,3.9) ;
\draw[blue, thick, -> ] (3,2.5) -- (3,3.9) ;

\draw[blue, thick, -> ] (0,2.5) -- (0,5) ;
\draw[blue, thick, -> ] (1.5,4) -- (1.5,5) ;
\draw[blue, thick, -> ] (3,4) -- (3,5) ;

\draw (0,1) node[black, anchor=east] {$A^{(1)}$ \hspace{ 1em }} ;
\filldraw[gray] (0,0.95) rectangle (1.5,1.05) ;
\filldraw[gray] (0,1) circle (0.1);
\filldraw[gray] (1.5,1) circle (0.1);

\draw (0,2.5) node[black, anchor=east] {$A^{(3)}$ \hspace{ 1em }} ;
\filldraw[gray] (0,2.45) rectangle (3,2.55) ;
\filldraw[gray] (0,2.5) circle (0.1);
\filldraw[gray] (3,2.5) circle (0.1);

\draw (0,4) node[black, anchor=east] {$A^{(5)}$ \hspace{ 1em }} ;
\filldraw[gray] (1.5,3.95) rectangle (3,4.05) ;
\filldraw[gray] (1.5,4) circle (0.1);
\filldraw[gray] (3,4) circle (0.1);

\draw (0,1.7) node[anchor=west] {$\phi_{13}$};

\draw (1.5,1.7) node[anchor=west] {$\phi_{15}$};

\draw (3,3.2) node[anchor=west] {$\phi_{35}$};

\draw (0,0.5) node[anchor=west] {$\phi_{21}$} ;
\draw (1.5,0.5) node[anchor=west] {$\phi_{41}$} ;
\draw (3,0.5) node[anchor=west] {$\phi_{43}$} ;

\draw (0,4.5) node[anchor=west] {$\phi_{32}$} ;
\draw (1.5,4.5) node[anchor=west] {$\phi_{52}$} ;
\draw (3,4.5) node[anchor=west] {$\phi_{54}$} ;

\draw (4.75,2.6) node {\huge =};

\draw[blue, thick, -> ] (6.5,0) -- (6.5,3.4) ;
\draw[blue, thick, -> ] (8,0) -- (8,1.4) ;
\draw[blue, thick, -> ] (9.5,0) -- (9.5,1.4) ;

\draw[blue, thick, -> ] (6.5,3.5) -- (6.5,5) ;
\draw[blue, thick, -> ] (8,1.5) -- (8,3.4) ;
\draw[blue, thick, -> ] (8,3.5) -- (8,5) ;
\draw[blue, thick, -> ] (9.5,1.5) -- (9.5,5) ;

\draw (9.5,1.5) node[black, anchor=west] {\hspace{ .15em } $A^{(4)}$} ;
\filldraw[gray] (8,1.45) rectangle (9.5,1.55) ;
\filldraw[gray] (8,1.5) circle (0.1);
\filldraw[gray] (9.5,1.5) circle (0.1);

\draw (9.5,3.5) node[black, anchor=west] {\hspace{ .15em } $A^{(2)}$} ;
\filldraw[gray] (6.5,3.45) rectangle (8,3.55) ;
\filldraw[gray] (6.5,3.5) circle (0.1);
\filldraw[gray] (8,3.5) circle (0.1);

\draw (8,2.4) node[anchor=west] {$\phi_{24}$};

\draw (6.5,0.5) node[anchor=west] {$\phi_{21}$} ;
\draw (8,0.5) node[anchor=west] {$\phi_{41}$} ;
\draw (9.5,0.5) node[anchor=west] {$\phi_{43}$} ;

\draw (6.5,4.5) node[anchor=west] {$\phi_{32}$} ;
\draw (8,4.5) node[anchor=west] {$\phi_{52}$} ;
\draw (9.5,4.5) node[anchor=west] {$\phi_{54}$} ;

\end{tikzpicture}
   \caption{Multivectors $\phi^{i,j}$ as row entries for pentagon equation}
   \label{fig:penta-phi}
  \end{center}
 \end{figure}
\end{example}

\begin{theorem}\label{th:polygon}
$(2n+1)$-gon equation~\eqref{(2n+1)-gon} is satisfied by matrices~$A^{(q)}$~\eqref{A-def}.
\end{theorem}

\begin{proof}
Draw the diagrammatic representation of the $(2n+1)$-gon equation, and write a $\phi^{i,j}$ at each edge as explained above. Then it turns out that:
\begin{itemize}\itemsep 0pt
 \item the $n(n+1)/2$ input $\phi$'s are the same for the l.h.s. and r.h.s.,
 \item they are linearly independent due to Proposition~\eqref{p:phi-plgn},
 \item and l.h.s. and r.h.s. gave the same output $\phi$'s.
\end{itemize} 
Hence, $\mathrm{l.h.s.} = \mathrm{r.h.s.}$
\end{proof}

It follows, of course, that \eqref{gon-inv} is satisfied by~$B^{(q)}$~\eqref{B-def}.

\begin{theorem}\label{th:simplex}
$2n$-simplex equation is satisfied by $R$-matrices~\eqref{RABP}, with $A^{(q)}$ as in~\eqref{A-def} and~$B^{(q)}$ as in~\eqref{B-def}.
\end{theorem}

\begin{proof}
First, we note that our matrices~$R^{(q)}$ of type~\eqref{RABP}, and with $B^{(q)} = (A^{(q)})^{-1}$, are obviously equal to their inverses:
\[
R^{(q)} = (R^{(q)})^{-1}.
\]
It follows immediately that the r.h.s.\ of the $2n$-simplex equation~\eqref{2n-simplex} is equal to $(\text{l.h.s.})^{-1}$. We have thus to prove that
\begin{equation}\label{chnot2}
\text{l.h.s.} = (\text{l.h.s.})^{-1}.
\end{equation}

Moreover, as we noted after Proposition~\ref{p:d}, the l.h.s.\ is a direct sum of matrices corresponding, first, to the red and blue sectors, and second, to the green sector. And there is no problem with blue and red sectors: they give, essentially, the same $(2n+1)$-gons already proved in Theorem~\ref{th:polygon} (alternatively, \eqref{chnot2} in the blue and red sectors is clear from the block matrix form~\eqref{br-block}). It remains thus to prove~\eqref{chnot2} only for the green sector submatrix.

First, we consider the case where our field~$F$ is \emph{not of characteristic~2}: $\mathop{\mathrm{char}} F\ne 2$. In this case, we will prove that the l.h.s.\ in the green sector is a diagonalizable matrix with eigenvalues only~$\pm 1$ (which obviously implies~\eqref{chnot2}).

We will do it as follows: first, using \emph{row} vectors made of our~$\phi $'s~\eqref{phi}, we will show the existence of \emph{at least} $\frac{n(n+1)}{2}$-dimen\-sional row eigenspace with eigenvalue~$+1$, then, using \emph{column} vectors made of our~$\psi $'s~\eqref{psi}, we will show the existence of \emph{at least} $\frac{n(n-1)}{2}$-dimen\-sional row eigenspace with eigenvalue~$-1$. This will clearly give us what we want, because together we get
\[
\frac{n(n+1)}{2} + \frac{n(n-1)}{2} = n^2,
\]
which is, as one can easily check, the full dimension of the green sector. Moreover, this will show of course that the dimensions of the mentioned eigenspaces are \emph{exactly} $\frac{n(n+1)}{2}$ and~$\frac{n(n-1)}{2}$.

First, we put the $\phi $'s on edges in the same manner as in the proof of Theorem~\ref{th:polygon}, and refer to Proposition~\ref{p:green}~$\mathrm{(i)}$. This gives (at least) a $n(n+1)/2$-dimensional linear row eigen(sub)space to~$+1$, as desired.
 
Second, we put the $\psi $'s on edges in almost the same manner as $\phi $'s, but adding minuses where needed. This is because of different signs in the right-hand sides of \eqref{psi-A} and~\eqref{psi-B} compared to \eqref{phi-A} and~\eqref{phi-B}. Namely, we put minuses at the input row entries of~$R^{(j)}$ with an odd~$j$, or/and output row entries of~$R^{(k)}$ with an even~$k$, and pluses otherwise. As one can check, this brings no contradictions because of Proposition~\ref{p:g}~$\mathrm{(ii)}$, and the signs agree with \eqref{psi-A} and~\eqref{psi-B}. Hence, we have, according to Proposition~\ref{p:green}~$\mathrm{(ii)}$, (at least) an $n(n-1)/2$-dimensional linear column eigen(sub)space to~$-1$.

Finally, we note that the case $\mathop{\mathrm{char}} F = 2$ follows easily from the case $\mathop{\mathrm{char}} F = \nobreak 0$. Namely, as soon as we have proved that $R$-matrices made of matrices \eqref{A-def} and~\eqref{B-def} satisfy the $2n$-simplex equation~\eqref{2n-simplex} in characteristic~0, and taking into account that \eqref{2n-simplex} can be re-written as a system of polynomial identities with integer coefficients (and matrix~$\mathcal M$~\eqref{M} entries as variables), we obtain \eqref{2n-simplex} in characteristic~2 simply by the reduction modulo~2 from characteristic~0.
\end{proof}

\begin{example}
For illustration of how the $\phi$'s and~$\psi$'s are placed on the l.h.s.\ of the 4-simplex equation diagram, see Figures \ref{fig:phis} and~\ref{fig:psis}.
 \begin{figure}
  \begin{center}
   \begin{tikzpicture}

\draw[green, thick] (1,0) -- (1,0.5) node[black, anchor=west] {$\phi^{3,1}$} -- (1,1) ;
\draw[green, thick] (3,0) -- (3,0.5) node[black, anchor=west] {$\phi^{5,1}$} -- (3,1) ;

\draw[green, thick] (5,0) -- (5,0.5) node[black, anchor=west] {$\phi^{4,2}$} -- (5,2) ;
\draw[green, thick] (8,0) -- (8,0.5) node[black, anchor=west] {$\phi^{5,3}$} -- (8,3) ;

\draw[green, thick] (0,1) -- (0,1.5) node[black, anchor=west] {$\phi^{1,2}$} -- (0,2) ;
\draw[green, thick] (2,1) -- (2,1.5) node[black, anchor=west] {$\phi^{1,4}$} -- (2,4) ;

\draw[green, thick] (4,2) -- (4,2.5) node[black, anchor=west] {$\phi^{2,3}$} -- (4,3) ;
\draw[green, thick] (6,2) -- (6,2.5) node[black, anchor=west] {$\phi^{2,5}$} -- (6,5) ;

\draw[green, thick] (1,3) -- (1,5.5) node[black, anchor=west] {$\phi^{3,1}$} -- (1,6) ;
\draw[green, thick] (7,3) -- (7,3.5) node[black, anchor=west] {$\phi^{3,4}$} -- (7,4) ;

\draw[green, thick] (5,4) -- (5,5.5) node[black, anchor=west] {$\phi^{4,2}$} -- (5,6) ;
\draw[green, thick] (9,4) -- (9,4.5) node[black, anchor=west] {$\phi^{4,5}$} -- (9,5) ;

\draw[green, thick] (3,5) -- (3,5.5) node[black, anchor=west] {$\phi^{5,1}$} -- (3,6) ;
\draw[green, thick] (8,5) -- (8,5.5) node[black, anchor=west] {$\phi^{5,3}$} -- (8,6) ;

\draw (0,1) node[black, anchor=east] {$R^{(1)} \hspace{ 1em }$} ;
\filldraw[gray] (0,0.95) rectangle (3,1.05) ;
\filldraw[gray] (0,1) circle (0.1);
\filldraw[gray] (1,1) circle (0.1);
\filldraw[gray] (2,1) circle (0.1);
\filldraw[gray] (3,1) circle (0.1);

\draw (0,2) node[black, anchor=east] {$R^{(2)} \hspace{ 1em }$} ;
\filldraw[gray] (0,1.95) rectangle (6,2.05) ;
\filldraw[gray] (0,2) circle (0.1);
\filldraw[gray] (4,2) circle (0.1);
\filldraw[gray] (5,2) circle (0.1);
\filldraw[gray] (6,2) circle (0.1);

\draw (0,3) node[black, anchor=east] {$R^{(3)} \hspace{ 1em }$} ;
\filldraw[gray] (1,2.95) rectangle (8,3.05) ;
\filldraw[gray] (1,3) circle (0.1);
\filldraw[gray] (4,3) circle (0.1);
\filldraw[gray] (7,3) circle (0.1);
\filldraw[gray] (8,3) circle (0.1);

\draw (0,4) node[black, anchor=east] {$R^{(4)} \hspace{ 1em }$} ;
\filldraw[gray] (2,3.95) rectangle (9,4.05) ;
\filldraw[gray] (2,4) circle (0.1);
\filldraw[gray] (5,4) circle (0.1);
\filldraw[gray] (7,4) circle (0.1);
\filldraw[gray] (9,4) circle (0.1);

\draw (0,5) node[black, anchor=east] {$R^{(5)} \hspace{ 1em }$} ;
\filldraw[gray] (3,4.95) rectangle (9,5.05) ;
\filldraw[gray] (3,5) circle (0.1);
\filldraw[gray] (6,5) circle (0.1);
\filldraw[gray] (8,5) circle (0.1);
\filldraw[gray] (9,5) circle (0.1);

\draw (0,-0.3) node[blue, anchor=center] {$12$} ;
\draw (1,-0.3) node[green, anchor=center] {$13$} ;
\draw (2,-0.3) node[blue, anchor=center] {$14$} ;
\draw (3,-0.3) node[green, anchor=center] {$15$} ;
\draw (4,-0.3) node[red, anchor=center] {$23$} ;
\draw (5,-0.3) node[green, anchor=center] {$24$} ;
\draw (6,-0.3) node[red, anchor=center] {$25$} ;
\draw (7,-0.3) node[blue, anchor=center] {$34$} ;
\draw (8,-0.3) node[green, anchor=center] {$35$} ;
\draw (9,-0.3) node[red, anchor=center] {$45$} ;

\draw (0,6.3) node[red, anchor=center] {$12$} ;
\draw (1,6.3) node[green, anchor=center] {$13$} ;
\draw (2,6.3) node[red, anchor=center] {$14$} ;
\draw (3,6.3) node[green, anchor=center] {$15$} ;
\draw (4,6.3) node[blue, anchor=center] {$23$} ;
\draw (5,6.3) node[green, anchor=center] {$24$} ;
\draw (6,6.3) node[blue, anchor=center] {$25$} ;
\draw (7,6.3) node[red, anchor=center] {$34$} ;
\draw (8,6.3) node[green, anchor=center] {$35$} ;
\draw (9,6.3) node[blue, anchor=center] {$45$} ;

\end{tikzpicture}
  \end{center}
 \caption{$\phi$'s in the green sector of 4-simplex equation, left-hand side}
 \label{fig:phis}

 \bigskip \bigskip

  \begin{center}
   \begin{tikzpicture}

\draw[green, thick] (1,0) -- (1,0.5) node[black, anchor=east] {$-\psi_{3,1}$} -- (1,1) ;
\draw[green, thick] (3,0) -- (3,0.5) node[black, anchor=east] {$-\psi_{5,1}$} -- (3,1) ;

\draw[green, thick] (5,0) -- (5,0.5) node[black, anchor=east] {$\psi_{4,2}$} -- (5,2) ;
\draw[green, thick] (8,0) -- (8,0.5) node[black, anchor=east] {$-\psi_{5,3}$} -- (8,3) ;

\draw[green, thick] (0,1) -- (0,1.5) node[black, anchor=east] {$\psi_{1,2}$} -- (0,2) ;
\draw[green, thick] (2,1) -- (2,1.5) node[black, anchor=east] {$\psi_{1,4}$} -- (2,4) ;

\draw[green, thick] (4,2) -- (4,2.5) node[black, anchor=east] {$-\psi_{2,3}$} -- (4,3) ;
\draw[green, thick] (6,2) -- (6,2.5) node[black, anchor=east] {$-\psi_{2,5}$} -- (6,5) ;

\draw[green, thick] (1,3) -- (1,5.5) node[black, anchor=east] {$\psi_{3,1}$} -- (1,6) ;
\draw[green, thick] (7,3) -- (7,3.5) node[black, anchor=east] {$\psi_{3,4}$} -- (7,4) ;

\draw[green, thick] (5,4) -- (5,5.5) node[black, anchor=east] {$-\psi_{4,2}$} -- (5,6) ;
\draw[green, thick] (9,4) -- (9,4.5) node[black, anchor=east] {$-\psi_{4,5}$} -- (9,5) ;

\draw[green, thick] (3,5) -- (3,5.5) node[black, anchor=east] {$\psi_{5,1}$} -- (3,6) ;
\draw[green, thick] (8,5) -- (8,5.5) node[black, anchor=east] {$\psi_{5,3}$} -- (8,6) ;

\draw (0,1) node[black, anchor=east] {$R^{(1)} \hspace{ 2em }$} ;
\filldraw[gray] (0,0.95) rectangle (3,1.05) ;
\filldraw[gray] (0,1) circle (0.1);
\filldraw[gray] (1,1) circle (0.1);
\filldraw[gray] (2,1) circle (0.1);
\filldraw[gray] (3,1) circle (0.1);

\draw (0,2) node[black, anchor=east] {$R^{(2)} \hspace{ 2em }$} ;
\filldraw[gray] (0,1.95) rectangle (6,2.05) ;
\filldraw[gray] (0,2) circle (0.1);
\filldraw[gray] (4,2) circle (0.1);
\filldraw[gray] (5,2) circle (0.1);
\filldraw[gray] (6,2) circle (0.1);

\draw (0,3) node[black, anchor=east] {$R^{(3)} \hspace{ 2em }$} ;
\filldraw[gray] (1,2.95) rectangle (8,3.05) ;
\filldraw[gray] (1,3) circle (0.1);
\filldraw[gray] (4,3) circle (0.1);
\filldraw[gray] (7,3) circle (0.1);
\filldraw[gray] (8,3) circle (0.1);

\draw (0,4) node[black, anchor=east] {$R^{(4)} \hspace{ 2em }$} ;
\filldraw[gray] (2,3.95) rectangle (9,4.05) ;
\filldraw[gray] (2,4) circle (0.1);
\filldraw[gray] (5,4) circle (0.1);
\filldraw[gray] (7,4) circle (0.1);
\filldraw[gray] (9,4) circle (0.1);

\draw (0,5) node[black, anchor=east] {$R^{(5)} \hspace{ 2em }$} ;
\filldraw[gray] (3,4.95) rectangle (9,5.05) ;
\filldraw[gray] (3,5) circle (0.1);
\filldraw[gray] (6,5) circle (0.1);
\filldraw[gray] (8,5) circle (0.1);
\filldraw[gray] (9,5) circle (0.1);

\draw (0,-0.3) node[blue, anchor=center] {$12$} ;
\draw (1,-0.3) node[green, anchor=center] {$13$} ;
\draw (2,-0.3) node[blue, anchor=center] {$14$} ;
\draw (3,-0.3) node[green, anchor=center] {$15$} ;
\draw (4,-0.3) node[red, anchor=center] {$23$} ;
\draw (5,-0.3) node[green, anchor=center] {$24$} ;
\draw (6,-0.3) node[red, anchor=center] {$25$} ;
\draw (7,-0.3) node[blue, anchor=center] {$34$} ;
\draw (8,-0.3) node[green, anchor=center] {$35$} ;
\draw (9,-0.3) node[red, anchor=center] {$45$} ;

\draw (0,6.3) node[red, anchor=center] {$12$} ;
\draw (1,6.3) node[green, anchor=center] {$13$} ;
\draw (2,6.3) node[red, anchor=center] {$14$} ;
\draw (3,6.3) node[green, anchor=center] {$15$} ;
\draw (4,6.3) node[blue, anchor=center] {$23$} ;
\draw (5,6.3) node[green, anchor=center] {$24$} ;
\draw (6,6.3) node[blue, anchor=center] {$25$} ;
\draw (7,6.3) node[red, anchor=center] {$34$} ;
\draw (8,6.3) node[green, anchor=center] {$35$} ;
\draw (9,6.3) node[blue, anchor=center] {$45$} ;

\end{tikzpicture}
  \end{center}
 \caption{$\psi$'s in the green sector of 4-simplex equation, left-hand side}
 \label{fig:psis}
 \end{figure}
\end{example}

\begin{example}\label{x:n=1}
It is instructive to examine the ``trivial'' case $n = 1$: the \emph{trigon} equation~\cite{DM-H} and the 2-simplex (Yang--Baxter) equation. Matrix~$\mathcal M$~\eqref{M} representing an element of Grassmannian $\Gr(2,3)$ is then
\[
\mathcal M =
 \begin{pmatrix} \alpha_1 & \alpha_2 & \alpha_3 \\
 \beta_1 & \beta_2 & \beta_3 \end{pmatrix},
\]
and we find
\[
A^{(1)} = -\frac{p_{1,3}}{p_{1,2}},\qquad
A^{(2)} = \frac{p_{2,3}}{p_{1,2}},\qquad
A^{(3)} = -\frac{p_{2,3}}{p_{1,3}};
\]
the trigon equation is
\[
A^{(1)} A^{(3)} = A^{(2)} .
\]
The analogue of Figure~\ref{fig:pentagon} looks as follows:
\[
\text{lhs:}\quad 12  \xrightarrow{A^{(1)}} 13  \xrightarrow{A^{(3)}} 23 ,\qquad\qquad \text{rhs:}\quad 12 \xrightarrow{A^{(2)}} 23.
\]
Formula~\eqref{RABP} is reduced to $R^{(q)} = A_1^{(q)} B_2^{(q)} P_{12}$, hence
\[
R^{(1)} = \begin{pmatrix} 0 & -\frac{p_{1,3}}{p_{1,2}} \\ -\frac{p_{1,2}}{p_{1,3}} & 0 \end{pmatrix}, \qquad
R^{(2)} = \begin{pmatrix} 0 &  \frac{p_{2,3}}{p_{1,2}} \\  \frac{p_{1,2}}{p_{2,3}} & 0 \end{pmatrix}, \qquad
R^{(3)} = \begin{pmatrix} 0 & -\frac{p_{2,3}}{p_{1,3}} \\ -\frac{p_{1,3}}{p_{2,3}} & 0 \end{pmatrix},
\]
which is a trivial solution of the \emph{entwining}~\cite{KP} Yang--Baxter equation~\eqref{YB}:
\begin{equation}\label{YB again}
R_{12}^{(1)} R_{13}^{(2)} R_{23}^{(3)} = R_{23}^{(3)} R_{13}^{(2)} R_{12}^{(1)}.
\end{equation}

On the other hand, we can interpret the upper index as a parameter. For example setting
\[
\mathcal M = \begin{pmatrix} 1 & -\mu & 0 \\ 0 & -\lambda & 1 \end{pmatrix}, \qquad
R(\lambda ) = \begin{pmatrix} 0 & \lambda \\ 1 / \lambda & 0 \end{pmatrix},
\]
we get
\[
R^{(1)} = R(\lambda ), \quad R^{(2)} = R(\lambda  \mu ), \quad R^{(3)} = R(\mu )
\]
and equation~\eqref{YB again} becomes
\[
R_{12} (\lambda ) R_{13} (\lambda \mu ) R_{23} (\mu ) = R_{23} (\mu ) R_{13} (\lambda \mu ) R_{12} (\lambda ).
\]
\end{example}

\section{Discussion}\label{s:d}

We have constructed new ``nonconstant'' solutions to polygon and simplex equations, in their ``direct-sum'' form. Our construction has been known earlier only for the simple case of pentagon~\cite{KS}. In the heptagon case, there are also different solutions~\cite{heptagon}, or at least their possible relations to the present paper is not yet known. Our results may be applied to both topology of piecewise linear manifolds and as a starting point in search for solutions of quantum equations to be used in mathematical physics. Also, it may be interesting to try and apply our nonconstant Yang--Baxter to constructing invariants of knots and knotted surfaces.

Concerning directions of further research, we can mention the following.

Our ``polygons'' considered here had always an odd number~$(2n+1)$ of ``vertices'', so we have left out of our consideration half of the possible equations. Among this half, there are certainly very interesting equations that deserve a separate research, such as a direct sum hexagon~\cite{K-nc,K-c} (note also a quantum version of hexagon in~\cite[Eq.~(21)]{Kash-simple}).

One more interesting direction of research may be searching for \emph{non-commutative} generalizations of our constructions using a division ring instead of our field~$F$, and quasideterminants and quasi-Pl\"ucker coordinates~\cite{GGRW}.

And of course the natural next step consists in studying the \emph{cohomology} of our solutions, as one more possible step towards constructing solutions to \emph{quantum} equations. In these latter, $A^{(q)}$ or~$R^{(q)}$ or their analogues are linear operators acting in the \emph{tensor product} of the corresponding spaces (identified of course with their tensor products with the identities in the lacking spaces). The simplest quantum solutions are obtained from set-theoretic solutions as follows: for each~$i$, consider the vector space~$V_i$ over a field~$K$ whose \emph{basis} is~$X_i$---that is, such $V_i$ consists of formal (finite) linear combinations
\[
\alpha x_i + \beta y_i + \dots, \quad \alpha, \beta,\ldots \in K, \quad x_i,y_i, \ldots \in X_i,
\]
and define, say, $R^{(q)}_{\mathrm{quantum}}$ for Yang--Baxter requiring that
\begin{equation}\label{YBq}
\text{if}\quad R^{(q)}_{ij}(x_i,y_j) = (z_i,t_j),\quad\text{then}\quad R^{(q)}_{\mathrm{quantum}} (x_i\otimes y_j)=(z_i\otimes t_j).
\end{equation}
Cohomology arises if we add \emph{multipliers} to~\eqref{YBq}, that is, set
\[
R^{(q)}_{\mathrm{quantum}} (x_i\otimes y_j) = c(x_i,y_j)(z_i\otimes t_j), \quad c(x_i,y_j)\in K^*,
\]
where $K^*$ is the multiplicative group of~$K$. A similar definition for other simplex or polygon equations can of course be also done; note that $K$ may be different from our field~$F$, for instance, one can take $K=\mathbb C$, while $F$ being a finite field. More about such cohomology theories can be found in~\cite{KST,K-nc,K-c}.

\bibliographystyle{plain}
\bibliography{paper.bib}

\end{document}